\documentclass[a4paper,UKenglish,cleveref, autoref, thm-restate]{lipics-v2021}

\graphicspath{{./figures/}}

\usepackage{tikz}
\usepackage{mathtools}
\usepackage{todonotes}
\usepackage{xspace}
\usepackage{subcaption}
\usepackage{enumitem}
\usepackage[strings]{underscore}

\usepackage{ifthen}
\newboolean{lipics}
\setboolean{lipics}{true}

\ifthenelse{\boolean{lipics}}{
  \pdfoutput=1 \hideLIPIcs  \bibliographystyle{plainurl}\title{Parameterized Complexity of Simultaneous Planarity}

  \author{Simon D. Fink}{Algorithms and Complexity Group, TU Wien, Austria}{sfink@ac.tuwien.ac.at}{https://orcid.org/0000-0002-2754-1195}{Funded by the Vienna Science and Technology Fund (WWTF) \newline[10.47379/ICT22029]}
  
  \author{Matthias Pfretzschner}{Faculty of Informatics and Mathematics, University of Passau, Germany}{pfretzschner@fim.uni-passau.de}{https://orcid.org/0000-0002-5378-1694}{Funded by the Deutsche Forschungsgemeinschaft (German Research Foundation, DFG) under grant RU-1903/3-1.}
  
  \author{Ignaz Rutter}{Faculty of Informatics and Mathematics, University of Passau, Germany}{rutter@fim.uni-passau.de}{https://orcid.org/0000-0002-3794-4406}{Funded by the Deutsche Forschungsgemeinschaft (German Research Foundation, DFG) under grant RU-1903/3-1.}
  
  \authorrunning{S.\,D. Fink and M. Pfretzschner and I. Rutter} 
  \Copyright{Simon D. Fink and Matthias Pfretzschner and Ignaz Rutter} 
  
  \category{} 
  \relatedversion{}

  \nolinenumbers

  \newtheorem{redrule}{Reduction Rule}[section]
}
{
\usepackage{graphicx}\usepackage{multirow}\usepackage{amsmath,amssymb,amsfonts}\usepackage{amsthm}\usepackage{mathrsfs}\usepackage[title]{appendix}\usepackage{xcolor}\usepackage{textcomp}\usepackage{manyfoot}\usepackage{booktabs}\usepackage{algorithm}\usepackage{algorithmicx}\usepackage{algpseudocode}\usepackage{listings}\usepackage{cleveref}
\usepackage{tikz}
\usepackage{mathtools}
\usepackage{todonotes}
\usepackage{xspace}
\usepackage{subcaption}
\usepackage{enumitem}

\newtheoremstyle{thmstyleone}{1em}{1em}{\itshape}{}{\bfseries}{.}{0.5em}{}\theoremstyle{thmstyleone}\newtheorem{theorem}{Theorem}\newtheorem{lemma}{Lemma}\newtheorem{corollary}{Corollary}\newtheorem{redrule}{Reduction Rule}
\theoremstyle{thmstyletwo}
\theoremstyle{thmstylethree}\newtheorem{definition}{Definition}
\raggedbottom

\title[Parameterized Complexity of Simultaneous Planarity]{Parameterized Complexity of Simultaneous Planarity}

\author[1]{\fnm{Simon D.} \sur{Fink}}\email{sfink@ac.tuwien.ac.at}

\author*[2]{\fnm{Matthias} \sur{Pfretzschner}}\email{pfretzschner@fim.uni-passau.de}

\author[2]{\fnm{Ignaz} \sur{Rutter}}\email{rutter@fim.uni-passau.de}

\affil[1]{\orgdiv{Algorithms and Complexity Group}, \orgname{TU Wien}, \city{Vienna}, \country{Austria}}

\affil[2]{\orgdiv{Faculty of Informatics and Mathematics}, \orgname{University of Passau}, \city{Passau}, \country{Germany}}


}

\Crefname{redrule}{Reduction Rule}{Reduction Rules}
\newcommand{\G}[1]{\ensuremath{\excl{G}{#1}}\xspace}
\newcommand{\mcE}{\ensuremath{\mathcal{E}}\xspace}
\newcommand{\SU}{\ensuremath{S^\cup}\xspace}
\renewcommand{\S}[1]{\ensuremath{\excl{S}{#1}}\xspace}
\newcommand{\mcEU}{\ensuremath{\mathcal{E}^\cup}\xspace}
\newcommand{\mcEexcl}[1]{\excl{\ensuremath{\mathcal{E}}}{#1}\xspace}
\newcommand{\hatmcEexcl}[1]{\excl{\ensuremath{\hat{\mathcal{E}}}}{#1}\xspace}
\newcommand{\GU}{\ensuremath{G^{\cup}}\xspace}
\newcommand{\hatGU}{\ensuremath{\hat{G}^{\cup}}\xspace}
\newcommand{\hatmcEU}{\ensuremath{\hat{\mathcal{E}}^\cup}\xspace}
\newcommand{\BO}[1]{\ensuremath{\mathcal{O}(#1)}}
\newcommand{\SEFE}{\textnormal{\textsc{SEFE}}\xspace}
\newcommand{\SEFEl}{\textnormal{\textsc{Simultaneous Embedding With Fixed Edges}}\xspace}

\newcommand{\Do}{\ensuremath{\Delta_{1}}\xspace}
\newcommand{\Di}{\ensuremath{\excl{\Delta}{i}}\xspace}
\newcommand{\DU}{\ensuremath{\Delta^\cup}\xspace}

\newcommand{\rulesep}{\unskip\ \vrule\ }
\DeclareRobustCommand*\circled[1]{\tikz[baseline=(char.base)]{
    \node[shape=circle,draw,inner sep=0.5pt] (char) {\ensuremath{#1}};}}
\DeclareRobustCommand{\excl}[2]{\ensuremath{#1^{\tikz[baseline=(char.base)]{
        \node[shape=circle,draw,inner sep=0.5pt] (char) {\tiny\ensuremath{#2}};}}}}
\DeclareMathOperator{\vc}{vc}
\DeclareMathOperator{\fes}{fes}
\DeclareMathOperator{\cv}{cv}
\DeclareMathOperator{\shc}{cc}

\begin{document}
  
  \maketitle
  
  \newcommand{\abstracttext}{Given $k$ input graphs $\G{1}, \dots ,\G{k}$, where each pair \G{i}, \G{j} with $i \neq j$ shares the same graph $G$,
    the problem \SEFEl (\SEFE) asks whether there exists a planar drawing for each input graph such that all drawings
    coincide on $G$. While \SEFE is still open for the case of two input graphs, the problem is NP-complete for $k \geq
    3$~\cite{Schaefer13}.
    
    In this work, we explore the parameterized complexity of \SEFE. We show that \SEFE is FPT
    with respect to $k$ plus the vertex cover number or the feedback edge set number of the union graph $\GU = \G{1} \cup
    \dots \cup \G{k}$.
    Regarding the shared graph~$G$, we show that \SEFE is NP-complete, even if $G$ is a tree with maximum degree 4.
    Together with a known NP-hardness reduction~\cite{Angelini15}, this
    allows us to conclude that several parameters of $G$, including the maximum degree, the maximum number of degree-1 neighbors,
    the vertex cover number, and the number of cutvertices are intractable. We also settle the tractability of all pairs of these
    parameters. We give FPT algorithms for the vertex cover number plus either of the first two parameters and for the number of
    cutvertices plus the maximum degree, whereas we prove all remaining combinations to be intractable.
  }
  
  \ifthenelse{\boolean{lipics}}
  {
    \begin{abstract}
      \abstracttext
    \end{abstract}
  }
  {
    \abstract{
      \abstracttext
    }
  }
  
  \section{Introduction}
  
  Let $\G1= (V, \excl{E}{1})$, \dots,  $\G{k} = (V, \excl{E}{k})$ denote $k$ graphs on the same vertex set, where each
  pair $(\G{i}, \G{j})$ has a common \emph{shared graph} $\G{i} \cap \G {j} = (V, \excl{E}{i} \cap \excl{E}{j})$. The
  problem \textsc{\SEFEl} (\SEFE) asks whether there exist planar drawings $\excl{\Gamma}{1}, \dots,
  \excl{\Gamma}{k}$ of $\excl{G}{1},\dots,\excl{G}{k}$, respectively, such that each pair $\excl{\Gamma}{i},
  \excl{\Gamma}{j}$ induces the same drawing on the shared graph $\G{i} \cap \G{j}$~\cite{Erten05}. 
  We refer to such a $k$-tuple of drawings as a \emph{simultaneous drawing}. Unless stated otherwise, we assume that~$k$
  is part of the input and not necessarily a constant. In this work, we focus on the restricted \emph{sunflower case} of
  \SEFE, where every pair of input graphs has the same shared graph~$G$. For more than two input graphs, \SEFE has been
  proven to be NP-complete \cite{Gassner06}, even in the sunflower case~\cite{Schaefer13}. For two input graphs, however,
  \SEFE remains open.

  \begin{figure}[t]
    \centering
    \begin{subfigure}[t]{\textwidth}
      \centering
      \includegraphics[scale=1,page=1]{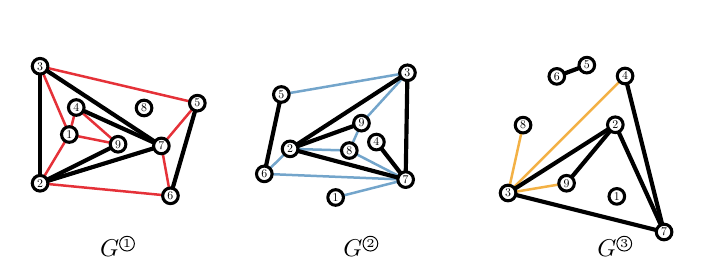}
      \caption{}
      \label{fig:ThreeSefeA}
    \end{subfigure}
    \vskip\baselineskip
    \begin{subfigure}[t]{\textwidth}
      \centering
      \includegraphics[scale=1,page=2]{figures/AnimationSteps}
      \caption{}
      \label{fig:ThreeSefeB}
    \end{subfigure}
    \caption{(a) Three planar graphs $\excl{G}{1}$, \excl{G}{2}, and \excl{G}{3} with the shared graph $G$ highlighted in
      black. (b) Planar drawings of the same three graphs, but the subgraph $G$ is drawn the same way in all three
      drawings.}
    \label{fig:ThreeSefe}
  \end{figure}
  
  One of the main applications for the \SEFE problem is dynamic graph drawing. Given a graph that changes over time, a
  visualization of $k$ individual snapshots of the graph should aesthetically display the changes between successive
  snapshots. 
  To this end, it is helpful to draw unchanged parts of the graph consistently.
  \Cref{fig:ThreeSefe} gives an example illustrating this for $k = 3$ snapshots.
  Using the same layout for the shared graph in \Cref{fig:ThreeSefeB} notably simplifies recognizing similarities and differences
  when compared to the varying layouts in \Cref{fig:ThreeSefeA}.

  In recent years, \SEFE received much attention and many algorithms solving restricted cases have been
  developed~\cite{Rutter20,BlasiusKR13}.
  Most notably, this includes the cases where every connected component of $G$ is biconnected~\cite{BlasiusKR18} or has a
  fixed embedding~\cite{Blaesius15}, and the case where $G$ has maximum degree~3~\cite{BlasiusKR18}. Very recently, Fulek
  and T{\'{o}}th \cite{Fulek20} solved \SEFE for $k = 2$ input graphs if the shared graph is connected and Bläsius et
  al.~\cite{Blaesius21B} later improved the running time from $O(n^{16})$ to quadratic using a reduction to the problem
  \textsc{Synchronized Planarity}. For $k \geq 3$ input graphs, however, the same restricted case remains
  NP-complete~\cite{Angelini15}. Despite the plethora of work dealing with \SEFE in restricted cases, we are not aware of
  parameterized approaches to \SEFE.
  \ifthenelse{\boolean{lipics}}
  {
  \begin{figure}[t]
    \centering
    \begin{minipage}{0.705\linewidth}
      \includegraphics[scale=1,page=1]{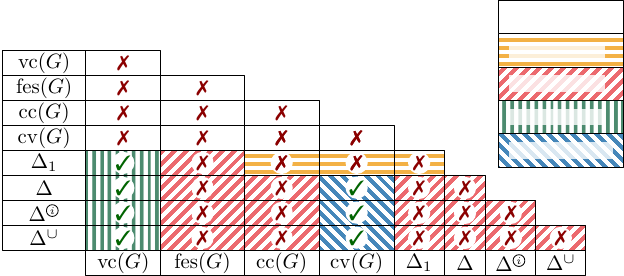}\end{minipage}\begin{minipage}{0.15\linewidth}
      \renewcommand{\arraystretch}{1.23}
      \hspace{-1.55cm}
      \begin{tabular}{l}
        \Cref{thm:hardStar} \\
        \Cref{thm:hardCvDo} \\
        \Cref{thm:unionmaxdeg4} \\
        \Cref{thm:vcDeg1} \\
        \Cref{thm:CvMaxDeg}
      \end{tabular}
      \vspace{1.83cm}
    \end{minipage}
    \caption{Complexity of \SEFE parameterized by combinations of the parameters of \Cref
      {sec:sharedParameters}, assuming that the number of input graphs $k$ is part of the input.
      Check marks indicate tractability and crosses indicate para-NP-hardness.
      If $k$ is fixed, then \Cref{thm:unionmaxdeg4} does not imply hardness, as the reduction requires an unbounded number of input graphs; the corresponding cells of the table are therefore still open.
      All other results also hold in this case.
}
    \label{fig:Table1}
  \end{figure}
  }
  {
  \begin{figure}[t]
    \centering
    \begin{minipage}{0.705\linewidth}
      \includegraphics[scale=1,page=1]{Tables}\end{minipage}\begin{minipage}{0.15\linewidth}
      \renewcommand{\arraystretch}{1.32}
      \hspace{-0.98cm}
      \begin{tabular}{l}
        \Cref{thm:hardStar} \\
        \Cref{thm:hardCvDo} \\
        \Cref{thm:unionmaxdeg4} \\
        \Cref{thm:vcDeg1} \\
        \Cref{thm:CvMaxDeg}
      \end{tabular}
      \vspace{1.8cm}
    \end{minipage}
    \caption{Complexity of \SEFE parameterized by combinations of the parameters of \Cref
      {sec:sharedParameters}, assuming that the number of input graphs $k$ is part of the input.
      Check marks indicate tractability and crosses indicate para-NP-hardness.
      If $k$ is fixed, then \Cref{thm:unionmaxdeg4} does not imply hardness, as the reduction requires an unbounded number of input graphs; the corresponding cells of the table are therefore still open.
      All other results also hold in this case.
}
    \label{fig:Table1}
  \end{figure}
  }
  In this work, we explore the parameterized complexity of \SEFE. In \Cref{sec:unionParameters}, we consider parameters of
  the \emph{union graph} $\GU = \G1 \cup \dots \cup \G{k}$. We show that \SEFE is FPT
  with respect to $k$ plus the vertex cover number or the feedback edge set number of \GU. Additionally, we prove intractability for \SEFE
  parameterized by the twin cover number of \GU.
In \Cref{sec:sharedParameters}, we turn to parameters of the shared graph~$G$. We consider as parameters the vertex cover
  number $\vc(G)$, the feedback edge set number $\fes(G)$, the number of connected components $\shc (G)$, the number of cutvertices
  $\cv(G)$, the maximum number of degree-1 neighbors~\Do in $G$, the maximum degree $\Delta$ of $G$, the maximum degree~\Di among all input
  graphs, and the maximum degree \DU of the union graph. For the latter four parameters, note that $\Do \leq \Delta
  \leq \Di \leq \DU$. \Cref{fig:Table1} gives an overview over these parameters together with their individual and
  pairwise tractabilities. In \Cref{sec:maxdeg}, we show that \SEFE is NP-complete even if the shared graph is a tree with
  maximum degree~4. This allows us to conclude that most parameters and their combinations are intractable. The only
  exceptions are due to the FPT-algorithms from \Cref{sec:shared-vc} for $\vc(G) + \Do$ and from \Cref{sec:cvMaxDeg} for
  $\cv(G) + \Delta$.
  
  \section{Preliminaries}
  
  For $k \in \mathbb{N}$, we define $[k] \coloneqq \{1,\dots, k\}$.
  Let $G = (V, E)$ be a simple graph. The \emph{open neighborhood} $N_G(v)$ of a vertex $v \in V$ denotes the set of vertices
  adjacent to $v$ in~$G$. The \emph{closed neighborhood} $N_G[v] \coloneqq N_G(v) \cup \{v\}$ additionally contains~$v$.
  If the graph is clear from context, we simply write $N(v)$ and $N[v]$, respectively. An \emph{induced subgraph} $H=(V',E')$ of~$G$
  contains all edges $E'\subseteq E$ with both endpoints in a given set~$V'\subseteq V$. The graph $G$ is \emph{connected},
  if, for every pair $u, v \in V$, there exists a path between $u$ and $v$ in $G$. A \emph{separating $k$-set} is a set $S
  \subseteq V$ with $|S| = k$ such that the graph $G - S$ obtained by removing $S$ is disconnected. The \emph{split
    components} of a separating $k$-set $S$ are the maximal subgraphs of~$G$ that are not disconnected by removing $S$. A
  separating $1$-set is also called a \emph{cutvertex}, a separating $2$-set is a \emph{separating pair}. We say that $G$
  is \emph{biconnected} if it contains no cutvertex and \emph{triconnected} if it contains no separating pair. A maximal
  induced subgraph of $G$ that is (bi-)connected is called a \emph{(bi-)connected component} of $G$. A
  biconnected component is also called a \emph{block}. A \emph{wheel} is a graph that consists of a cycle and an
  additional vertex adjacent to every vertex of the cycle.
  
  \subparagraph{Parameterized Complexity} 
  A parameterized problem $L \subseteq \Sigma^* \times \mathbb{N}$ is \emph{fixed-parameter tractable (FPT)}, if $L$ can
  be solved in time $f(k) \cdot n^{O(1)}$, where $f$ is some computable function and $k$ is the parameter. The
  problem $L$ is \emph{para-NP-hard} with respect to $k$ if
  $L$ is NP-hard even for constant values of $k$. For a graph $G = (V, E)$, a \emph{vertex cover} is a vertex set $C
  \subseteq V$ such that every edge $e \in E$ is incident to a vertex in~$C$. The \emph{vertex cover number} $\vc(G)$ is
  the size of a minimum vertex cover of $G$. The following lemma states that the vertex cover number of a planar graph gives an upper bound for the number of its high-degree vertices.
  
  \begin{lemma}
    \label{lm:planarDeg3}
    Let $G = (V, E)$ be a planar graph and let $N_3 \subseteq V$ denote the set of vertices of degree at least 3 in $G$.
    Then $|N_3| \leq 3 \cdot \vc(G)$.
  \end{lemma}
  \begin{proof}
    Let $C$ denote a minimum vertex cover of size $\vc(G)$ of $G$. Observe that all neighbors of a vertex in $V
    \setminus C$ must be contained in $C$, because otherwise, we immediately get an edge that is not covered by~$C$.
    Thus every vertex in $N_3 \cap (V \setminus C)$ has at least three neighbors in~$C$. Therefore, we can use planarity
    properties derived from Euler's Formula to infer that $|N_3 \cap (V \setminus C)| \leq \max(0, 2|C| - 4) \leq
    2|C|$ \cite[Lemma~13.3]{Fomin2019}. Together we get that
    \[|N_3| = |N_3 \cap C| + |N_3 \cap (V \setminus C)| \leq 3 \cdot \vc(G).\qedhere\]
  \end{proof}
  
  An edge $\{u, v\} \in E$ is a \emph{twin edge} if $N[u] = N[v]$, that is, $u$ and $v$ have the same neighborhood. A set
  $C \subseteq V$ is a \emph{twin cover} of $G$ if every edge $e \in E$ is a twin edge or incident to a vertex of $C$. The
  \emph{twin cover number} of $G$ is the size of a minimum twin cover of $G$.
A \emph{feedback edge set} of $G$ is an edge set $F \subseteq E$ such that $G - F$ is acyclic.
  The \emph{feedback edge set number} $\fes(G)$ is the minimum size of a feedback edge set of $G$.
  
  \subparagraph{SPQR-Trees} A pair $\{u, v\}$ of vertices is a \emph{split pair} of $G$ if 
  $\{u, v\}$ is a separating pair or a pair of adjacent vertices. An \emph{SPQR-tree}~\cite{Battista96}
  $\mathcal{T}$ is a tree that decomposes a biconnected graph~$G$ along its split pairs and can be computed in linear
  time~\cite{Gutwenger00}. The leaves of $\mathcal{T}$ are called \emph{Q-nodes} and correspond bijectively to the edges
  of $G$. Every inner node $v$ of $\mathcal{T}$ represents a biconnected multigraph (the \emph{skeleton} of $v$) that is
  either an \emph{S-node} (``\emph{series}'') consisting of a simple cycle, a \emph{P-node} (``\emph{parallel}'')
  consisting of two \emph{pole} vertices connected by at least three parallel edges, or an \emph{R-node} (``\emph{rigid}'') consisting of a triconnected simple graph. These skeletons consist of \emph{virtual edges}, where every
  virtual edge $\varepsilon$ represents a subgraph of~$G$, called the \emph{expansion graph} $\exp(\varepsilon)$ of
  $\varepsilon$.
  
  Fixing a planar embedding of $G$ completely fixes the embedding of the skeletons of all nodes in $\mathcal{T}$ and
  conversely, choosing a planar embedding for every skeleton in $\mathcal{T}$ uniquely defines a planar embedding of $G$.
  Therefore, all planar embedding decisions of a biconnected graph break down to ordering parallel edges in the skeletons of P-nodes and
  flipping the unique embedding of the skeleton of R-nodes.
  
  The \emph{SPQR-forest} of a graph is the collection of SPQR-trees of its biconnected components.

  \subparagraph{SEFE} The graph $\GU = (V, \bigcup_{i \in [k]}\excl{E}{i})$ consisting of the edges of all input graphs
  is called the \emph{union graph}.
  For brevity, we describe instances of \SEFE using the union graph by marking every edge of \GU  with the input graphs it is contained in.
  Every edge of \GU is either contained in exactly one input graph~\G{i}, or in all of them. In the former case, we
  say that the edge is \emph{$\circled{i}$-exclusive}, in the latter case it is \emph{shared}.
  The connected components of the shared graph are also called \emph{shared components}. 
  
  Jünger and Schulz~\cite{Juenger09} showed that an instance of \SEFE admits a simultaneous embedding if and only if there
  exist embeddings of the input graphs that are \emph{compatible}, i.e., they satisfy the following two requirements:
  (1) The cyclic ordering of the edges around every vertex of~$G$ must be identical in all embeddings.
  (2) For every pair $C$ and $C'$ of connected components in $G$, the face of $C$ that $C'$ is embedded in must be the same in all embeddings.
  We call the former property \emph{consistent edge orderings} and the latter property \emph{consistent relative positions}~\cite{Rutter20}.
  Note that any \SEFE instance with a non-planar input graph is a no-instance, we thus assume all input graphs to be planar.
  Furthermore, Bläsius et al.~\cite{BlasiusKR18} showed that one can assume that the union graph is biconnected.

  \section{Parameters of the Union Graph}
  \label{sec:unionParameters}
  
  In this section, we study the parameters vertex cover number, feedback edge set number, and twin cover number of the
  union graph $\GU = \G1 \cup \dots \cup \G{k}$. We give an FPT algorithm for each of the former two in combination with $k$ and show the latter to be
  intractable.
  
  \subsection{Vertex Cover Number}
  \label{sec:vc-union}
  
  For our first parameterization, we consider the vertex cover number of the union graph~$\vc
  (\GU)$. We use a similar approach as Bhore et al. \cite{Bhore2020} in their parameterization of the problem \textsc{Book
    Thickness}.
Let $C$ be a minimum vertex cover of \GU of size $\varphi \coloneq \vc(\GU)$ and let $k$ be the number of input graphs.
  For every vertex $v \in V \setminus C$, note that $N_{\GU}(v) \subseteq C$, as otherwise there would be an edge in $\GU$
  not covered by $C$. We group the vertices of $V (\GU) \setminus C$ into \emph{types} based on their neighborhood in all
  input graphs, i.e., two vertices $v_1$ and $v_2$ are of the same type if and only if $N_{\excl{G}{i}}(v_1)
  =N_{\excl{G}{i}}(v_2)$ for all $i \in [k]$; see \Cref{fig:VCUnion-types}. Let $\mathcal{P}$
  denote the partition of $V \setminus C$ based on the types of vertices.
Let $\mathcal{P}_{\geq3} \subseteq P$ be the types of $\mathcal{P}$ whose vertices have degree at least 3 in some input
  graph~$\G{i}$ and let $\mathcal{P}_{\leq2} \coloneq \mathcal{P} \setminus \mathcal{P}_{\geq3}$ denote the remaining
  types that have degree at most~2 in every input graph. Our first goal is to bound the number of vertices in
  $\mathcal{P}_ {\geq3}$. Since every input graph \excl{G}{i} is planar, we can use \Cref{lm:planarDeg3} to bound the
  number of vertices of degree at least 3 in $\G{i}$ linearly in $\vc(\G{i}) \leq \varphi$. The union graph
  \GU can therefore only contain $O(k\varphi)$ vertices that have degree at least 3 in some exclusive graph. Consequently,
  we obtain the upper bound $| \bigcup \mathcal{P}_{\geq3}| \in O(k\varphi)$.
  
  \begin{figure}[t]
    \centering
    \begin{subfigure}[t]{0.45\textwidth}
      \centering
      \includegraphics[scale=1,page=1]{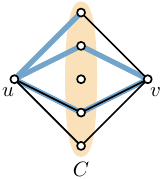}
      \caption{}
    \end{subfigure}
\hfill
    \begin{subfigure}[t]{0.45\textwidth}
      \centering
      \includegraphics[scale=1,page=2]{VCUnion-types}
      \caption{}
    \end{subfigure}
    \caption{(a) Vertices $u$ and $v$ of different types having different neighborhoods in the two input graphs.
      \excl{G}{1} is shown with thin black edges, \excl{G}{2} with thick blue edges, overlapping edges belong to~$G$.
      (b) Vertices $u$ and $v$ of the same~type.}
    \label{fig:VCUnion-types}
  \end{figure}
  
  It remains to bound the number of vertices contained in $\mathcal{P}_{\leq2}$. Since these vertices have degree at most
  2 in each input graph, there are $O(\binom{\varphi}{2}^k)$ distinct types in $\mathcal{P}_{\leq2}$, i.e., $|\mathcal{P}_
  {\leq2}| \in O(\varphi^{2k})$. To bound the number of vertices of each type in $\mathcal{P}_{\leq2}$, we use the
  following reduction rule; see \Cref{fig:RR-VCUnion}.
  \begin{figure}[t]
    \centering
    \includegraphics{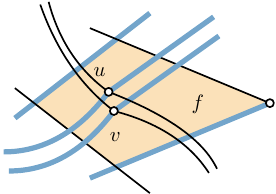}
    \caption{An example illustrating the proof of  \Cref{rr:vc-union} for $k = 2$ input graphs with vertices $u, v$ belonging to the same type of $\mathcal{P}_{\leq2}$. Vertex $v$ can be
      placed directly next to vertex $u$ in the face $f$ that $u$ is contained in without introducing any
      crossings in the input graphs.}
    \label{fig:RR-VCUnion}
  \end{figure}
  
  \begin{redrule}
    \label{rr:vc-union}
    If there exists a type $U \in \mathcal{P}_{\leq2}$ with $|U| > 1$, pick an arbitrary vertex $v \in U$ and reduce the
    instance to $(\GU-v, k)$.
  \end{redrule}
  
  \begin{proof}[Proof of Safeness]
    If the original instance admits a simultaneous drawing $\Gamma$, a simultaneous drawing of $\GU -v$ can be
    obtained by removing $v$ from $\Gamma$.
Conversely, let $\hat{\Gamma}^\cup = (\excl{\hat{\Gamma}}{1}, \dots, 
    \excl{\hat{\Gamma}}{k})$ be a simultaneous drawing of $\GU - v$. Pick an arbitrary vertex $u \in U \setminus \{v\}$ and let $f$ denote the face
    of $\hat{\Gamma}^\cup - u$ that $u$ is contained in. Since $u$ and $v$ have the same (at most) two neighbors in
    every input graph, $v$ can be placed close to $u$ in $f$ in $\hat{\Gamma}^\cup$ and each edge incident to $v$ can be
    routed directly next to the corresponding edge incident to $u$ without introducing crossings in the input graphs;
    see \Cref{fig:RR-VCUnion}.
    We therefore obtain a simultaneous drawing of~$\GU$.
  \end{proof}

  After exhaustively applying \Cref{rr:vc-union}, each type of  $\mathcal{P}_{\leq2}$ contains at most one vertex. Because
  we have $|\mathcal{P}_{\leq2}| \in O(\varphi^{2k})$, we thus get $|\bigcup \mathcal{P}_{\leq2}| \in O(\varphi^
  {2k})$. Using the upper bound $| \bigcup \mathcal{P}_{\geq3}| \in O(k\varphi)$ from above, we finally obtain a
  kernel of size $|V| = | \bigcup \mathcal{P}_{\geq3}| + |\bigcup \mathcal{P}_{\leq2}| + |C| \in O(\varphi^{2k})$ that can be solved in FPT time by enumerating all subsets of its vertices.
  Combined with the fact that a minimum vertex cover of
  size~$\varphi$ can be computed in time $O(1.2738^\varphi + \varphi n)$~\cite{Chen10}, we obtain the following result.
  
  \begin{theorem}
    \label{thm:vc-union}
    \SEFE is FPT with respect to the number of input graphs $k$ plus the vertex cover number $\varphi$ of the
    union graph \GU and admits an $ O(\varphi^{2k}) $ kernel.
  \end{theorem}
  We note that this result also holds in the non-sunflower case of~\SEFE.

  \subsection{Feedback Edge Set Number}
  \label{sec:fesunion}
  
  In this section, we consider the feedback edge set number $\psi = \fes(\GU)$ of the union graph.
  We build on ideas of Binucci et al.~\cite{Binucci24} in their parameterization of the problem \textsc{StoryPlan}. Given
  a minimum feedback edge set $F$ of \GU, our goal is to bound the number of vertices of $H \coloneqq \GU - F$ using reduction rules. Since $F$ is minimal and we can assume \GU to be
  connected, $H$ must be a tree.
  
  Recall that we can even assume the union graph \GU to be biconnected due to the preprocessing by Bläsius et
  al.~\cite{BlasiusKR18}. Since the preprocessing simply decomposes
  split components around cutvertices of the union graph into independent instances, the feedback edge set number of the
  graph does not increase. The biconnectivity of \GU ensures that \GU does not contain degree-1 vertices and therefore
  every leaf of $H$ is incident to an edge of $F$ in \GU. This allows us to bound the number of leaves (and
  consequently also the number of nodes of degree at least~3) of $H$ linearly in $|F| = \psi$.
  
  It thus only remains to limit the number of degree-2 vertices in $H$. 
  Let an \emph{$l$-chain} of \GU denote a maximal
  path consisting of $l+2$ vertices, where each of its $l \geq 1$ inner vertices has degree~2. Note that every degree-2
  vertex of a graph is contained in exactly one of its $l$-chains. To bound the number of degree-2 vertices in $H$, we
  first restrict the length of all $l$-chains in \GU. 

  \begin{definition}
    \label{def:redundant}{}
    An edge $e$ contained in an $l$-chain $c$ of \GU is called \emph{redundant} if
$e$ is a shared edge, or if
      $e$ is an \circled{i}-exclusive edge for some $i \in [k]$, such that $c$ also contains a different 
      \circled{i}-exclusive edge $e'$.
\end{definition}
  
  With the following lemma, we show that redundant edges contained in an $l$-chain of \GU can be safely contracted as
  illustrated in \Cref{fig:FESUnion}.

  \begin{figure}[t]
    \centering
    \begin{subfigure}[t]{0.45\textwidth}
      \centering
      \includegraphics[scale=1,page=1]{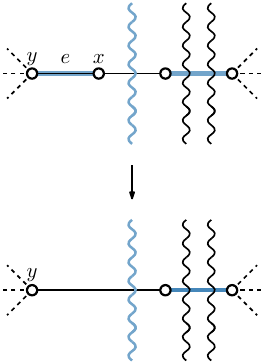}
      \caption{}
    \end{subfigure}
    \hfill
    \begin{subfigure}[t]{0.45\textwidth}
      \centering
      \includegraphics[scale=1,page=2]{RR-FES}
      \caption{}
    \end{subfigure}
    \caption{Contraction of a shared (a) or a \circled{1}-exclusive (b) edge $e$ in a $2$-chain.
      Because the chain contains an additional \circled{1}-exclusive edge $e'$ in (b), the endpoints $x$
      and $y$ of $e$ are contained in different connected components of all other input graphs. The
      \circled{2}-exclusive edge $g$ can therefore be rerouted as illustrated.}
    \label{fig:FESUnion}
  \end{figure}
  \begin{lemma}
    \label{lm:redundant}
    Let $e$ be a redundant edge of \GU and let $\hatGU$ denote the graph obtained from \GU by contracting $e$. Then
    \GU admits a simultaneous embedding if and only if $\hatGU$ does.
  \end{lemma}
  \begin{proof}
    Let $c$ denote the unique $l$-chain of \GU that $e$ is contained in and let $h$ be an edge of~$c$ connected to $e$
    via a degree-2 vertex $x$. Given a simultaneous drawing $\hat{\Gamma}$ of~$\hatGU$, subdivide edge $h$ once and
    change the type of the new edge $h_x$ incident to $x$ to the type of $e$ (i.e., either \circled{i}-exclusive or
    shared). Since $h_x$ can be drawn arbitrarily
    short, we can ensure that no edge crosses $h_x$ and we thus obtain a simultaneous drawing of \GU.
    
    To show the other direction, assume that \GU admits a simultaneous embedding $\mcEU = (\mcEexcl{1}, \dots,
    \mcEexcl{k})$ and let $\hatmcEU = (\hatmcEexcl{1}, \dots, \hatmcEexcl{k})$ denote the embedding obtained after
    contracting edge $e$ in \GU.
    
    We first show that every embedding $\hatmcEexcl{j}$ is planar for $j \in [k]$. To this end, we need to differentiate
    whether $e$ is contained in \G{j} or not. If $e$ is contained in \G{j}, then contracting $e$ in \GU is also a simple
    edge contraction in \G{j}, which retains the planarity of the embedding~\mcEexcl{j}. If $e$ is not contained in \G{j},
    then $e$ must be $\circled{i}$-exclusive for some $i \neq j$. The endpoints of $e$ are thus non-adjacent in \G{j} and
    contracting $e$ in \GU could lose the planarity of \mcEexcl{j} in the general case. However, by \Cref{def:redundant},
    $c$ contains a different \circled{i}-exclusive edge $e'\neq e$, thus the two endpoints of $e'$ are not connected in~\G{j}. Since $e$ and $e'$ are part of  the same $l$-chain, this means that the endpoints of~$e$ belong to different
    connected components in \G{j}, which guarantees planarity of
    \hatmcEexcl{j} after the contraction; see \Cref{fig:FESUnion}.
    
    Note that the contraction does not affect the consistency of edge orderings between the exclusive embeddings, it
    therefore only remains to argue the consistency of relative positions. Contracting a single edge $e$ in an $l$-chain
    $c$ of \GU can only violate the consistency of relative positions if $e$ is exclusive to some input graph and all
    other edges of $c$ are shared (otherwise the faces of the shared graph are unaffected). In this case, all edges of $c$ are incident to a single face of the shared graph before
    the contraction, but possibly incident to two separate faces after the contraction. However, by \Cref{def:redundant},
    such an edge $e$ is not redundant. This problematic case can therefore not occur and thus $\hatmcEU$ also has
    consistent relative positions and is consequently a simultaneous embedding of $\hatGU$.
  \end{proof}
  
  Using \Cref{lm:redundant}, we obtain the following reduction rule, allowing us to shorten~$l$-chains of \GU and thus
  reducing the number of degree-2 vertices.
  \begin{redrule}
    \label{rr:union-fes-chains}
    If $\GU$ contains a redundant edge $e$, then contract $e$ to obtain the graph $\hatGU$ and reduce the instance
    to $(\hatGU, \psi)$
  \end{redrule}
  \begin{proof}[Proof of Safeness]
    By \Cref{lm:redundant}, \GU and $\hatGU$ are equivalent.
If $e$ is contained in a minimum feedback edge set $F$, replace $e$ with any other edge of the $l$-chain $c$ that
    $e$ is contained in. Observe that this yields a minimum feedback edge set of the graph $\hatGU$, thus the
    parameter~$\psi$ does not change.
  \end{proof}
  
  After exhaustively applying \Cref{rr:union-fes-chains}, every $l$-chain $c$ of \GU with $l \geq 2$ contains no shared
  edge and no duplicate \circled{i}-exclusive edges. Since every edge of \GU is either shared or
  \circled{i}-exclusive, $c$ contains at most $k$ edges. Therefore, we have $l \leq k - 1$ for every
  $l$-chain in \GU. This yields an upper bound for the number of degree-2 vertices in $H$, resulting in a kernel for
  the problem.
  \begin{theorem}
    \SEFE is FPT with respect to the number of input graphs $k$ plus the feedback edge set number $\psi$
    of the union graph \GU and admits a kernel of size $O(k\psi)$.
  \end{theorem}
  \begin{proof}
    Let \GU denote the instance obtained after exhaustively applying \Cref{rr:union-fes-chains} and let $F$ denote a
    minimum feedback set of \GU. Because the reduction rule does not affect the parameter, we have $|F| = \psi$. Let $H
    \coloneqq \GU - F$ be the tree obtained by removing the edges in $F$ from \GU. For any leaf $v$ of $H$, observe that
    $v$ must be incident to an edge of $F$ in~\GU, because otherwise \GU would contain a degree-1 vertex, contradicting
    its biconnectivity. Because an edge of $F$ can be incident to at most two such leaves in \GU, it follows that $H$
    overall has at most $O(\psi)$ leaves. Since $H$ is a tree, it also contains at most $O(\psi)$ inner vertices of degree
    3 or higher. It therefore remains to bound the number of degree-2 vertices, which are all contained in some $l$-chain
    of $H$. Consider an $l$-chain $c$ of $H$ with $l \geq k$. Because for every $l'$-chain in \GU it is $l' \leq k - 1$,
    there must be an edge of $F$ incident to one of the degree-2 vertices in $c$ in \GU. More specifically, at least
    $\lfloor \frac{l}{k} \rfloor$ degree-2 vertices in $c$ must be incident to an edge of~$F$. The two endpoints of every
    edge in $F$ can therefore each ``pay'' for at most $k-1$ additional vertices in an $l$-chain. The number of degree-2
    vertices of $H$ is consequently in $O(k\psi)$. Overall, we thus obtain a kernel of size $O(k\psi)$.

    Since a minimum feedback edge set of an undirected graph can be computed in linear time by computing an arbitrary
    spanning tree and since the kernel can be solved by enumerating all subsets of its edges, it immediately follows that \SEFE is FPT parameterized by $k$ plus the feedback edge set number $\psi$
    of \GU.
  \end{proof}

  \subsection{Twin Cover Number}
  \label{sec:twincoverunion}
  
  Finally, we show that \SEFE is para-NP-hard with respect to the twin cover number of the union graph \GU.
  To this end, we use the following theorem.
  
  \begin{theorem}
    When $k$ is part of the input, \SEFE is NP-complete, even if the union graph is a complete graph.
  \end{theorem}
  \begin{proof}
    We start with the instance \GU obtained from the reduction by Angelini et al.~\cite{Angelini15} where the shared
    graph is a star (the remaining structure of the instance is irrelevant). For every edge $e$ of the complement graph
    of \GU (i.e., for every edge that is missing in \GU), create a new input graph \G{e} and add $e$ to \G{e}. Note that
    $\G{e}$ only consists of the shared star and the additional edge $e$. As this does not
    restrict the rotation of the star, we thus obtain an equivalent instance.
  \end{proof}
  
  Note that, if \GU is a complete graph, all vertices have the same neighborhood and thus every edge is a twin edge,
  and consequently \GU has twin cover number~0.
  
  \begin{theorem}
    \label{thm:twincoverhard}
    When $k$ is part of the input, \SEFE is para-NP-hard with respect to the twin cover number~of~\GU.
  \end{theorem}	
  
  Recall that the FPT algorithms from \Cref{sec:vc-union,sec:fesunion} also require the number of input graphs $k$ as a parameter.
  In contrast, \Cref{thm:twincoverhard} only holds for unbounded~$k$, since the reduction requires many input graphs.

  \section{Parameters of the Shared Graph}
  \label{sec:sharedParameters}
We now consider
  parameters of the shared graph $G$. In this case, finding safe reduction rules becomes significantly more involved.
  While we may assume that the union graph is biconnected~\cite{BlasiusKR18}, even isolated vertices or vertices of
  degree~1 of the shared graph may hold important information due to their connectivity in the union graph
  as the following result of Angelini et al.~\cite{Angelini15} shows.
\begin{theorem}[\cite{Angelini15}]
    \label{thm:hardStar}
    \SEFE is NP-complete for any fixed $k \geq 3$, even if $G$ is a star.
  \end{theorem}
  
  Since, for a star, the vertex cover number $\vc(G)$, the feedback edge set number $\fes(G)$, the number of connected
  components $\shc(G)$, and the number of cutvertices $\cv(G)$ are all constant, \Cref{thm:hardStar} already implies
  para-NP-hardness for each of these parameters and for all combinations of them; see \Cref{fig:Table1}. To obtain
  additional hardness results, one can take the instance resulting from \Cref{thm:hardStar}, duplicate every edge of the
  star, and subdivide each of the new duplication edges once. Since these new edges can be drawn directly next to their
  original version (or simply be removed from the drawing, if one wants to argue the reverse direction), the resulting
  instance is equivalent. In this instance, the parameter~$\Do$, the maximum number of degree-1 vertices adjacent to a
  single vertex in $G$, is zero, and the only cutvertex of $G$ is the center vertex, thus filling
  additional gaps in \Cref{fig:Table1}. We note that this
  construction increases $\fes(G)$ and~$\vc(G)$.
  \begin{theorem}
    \label{thm:hardCvDo}
    \SEFE is NP-complete for any fixed $k \geq 3$, even if $\Do = 0$ and $\cv(G) = cc(G) = 1$.
  \end{theorem}

  In the remainder of this section, we focus on several degree-related parameters for \SEFE. In addition to
  \Do, we consider the maximum degree $\Delta$ of the shared graph, the highest maximum degree
  $\Di$ among all input graphs, and the maximum degree $\DU$ of the union graph. Recall that $\Do \leq \Delta \leq \Di
  \leq \DU$. \Cref{thm:hardStar,thm:hardCvDo} do not prove hardness for the latter three parameters as the shared graph
  has high degree. To close this gap, we show in \Cref{sec:maxdeg} that \SEFE is
  NP-complete, even if the shared graph is a tree and the union graph has maximum degree 4. This
  proves the intractability of many combinations of parameters; see \Cref{fig:Table1}. Finally, we show that
  \SEFE is FPT when parameterized by $\vc(G) + \Do$ and by $\cv(G) + \Delta$ (\Cref{sec:shared-vc,sec:cvMaxDeg}), which settles
  the remaining entries~of~\Cref{fig:Table1}.

  \subsection{Maximum Degree}
  \label{sec:maxdeg}
  
  In this section we show that \SEFE is NP-complete, even if the instance has bounded degree. Angelini et
  al.~\cite{Angelini15} reduce from the NP-complete~\cite{Opatrny79} problem \textsc{Betweenness}, which asks for a linear
  ordering of a ground set $X$, subject to a set $\mathcal{T}$ of triplets over~$X$, where each triplet requires an
  element to appear between two other elements in the ordering. They use the rotation of a star in the shared graph to
  encode the linear ordering in the \SEFE instance, while exclusive edges enforce the triplets. Since this
  approach leads to a high-degree vertex, we use a different approach. We use small-degree vertices in the shared graph to
  encode the ordering of every three-element subset of $X$, while the input graphs ensure compatibility
  between these orderings. This results in a shared graph with small maximum degree, but the number of input graphs is
  not~constant. \begin{figure}[t]
    \centering
    \includegraphics{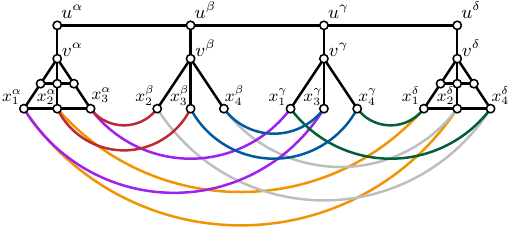}
    \caption{An example illustrating the reduction of \Cref{thm:npcMaxDeg4} with $X = \{x_1, x_2, x_3, x_4\}$ and
      $\mathcal{T} = \{(x_1, x_2, x_3), (x_1, x_2, x_4)\}$. Shared edges are drawn in black, every other color represents
      one input graph.}
    \label{fig:boundedDegree}
  \end{figure}
  \begin{theorem}
    \label{thm:npcMaxDeg4}
    When $k$ is part of the input, \SEFE is NP-complete, even if the shared graph is connected and has maximum degree 4.
  \end{theorem}
  \begin{proof}
    We reduce from the NP-complete~\cite{Opatrny79} problem \textsc{Betweenness}. Given a ground set $X = \{x_1,
    \dots, x_n\}$ and a set $\mathcal{T}$ of triplets over $X$, the problem asks whether there exists a linear ordering
    $\sigma$ of $X$, such that, for every triplet ${(a, b, c) \in \mathcal{T}}$, $b$ appears between $a$ and $c$
    in~$\sigma$.
    
    Given the instance
    ($X$, $\mathcal{T})$ of \textsc{Betweenness}, we construct an equivalent \SEFE instance~\GU as follows; see
    \Cref{fig:boundedDegree}.
    For every three-element subset \mbox{$\tau = \{x_i, x_j, x_k\} \in \binom{X}{3}$}, create a star in the shared graph $G$
    consisting of the center vertex $v^\tau$ adjacent to vertices $u^\tau$, $x_i^\tau$, $x_j^\tau$, and $x_k^\tau$; the
    latter three are called the \emph{element vertices}. Pick an arbitrary linear order $\alpha, \beta, \gamma, \dots$ for
    the elements of $\binom{X}{3}$ and add shared edges to build a path $u^\alpha, u^\beta, u^\gamma \dots$ in $G$.
    
    For every triplet $(x_i, x_j, x_k) \in \mathcal{T}$, let $\tau = \{x_i, x_j, x_k\} \in \binom{X}{3}$ denote the
    corresponding subset of $X$. Add the shared edges $x_i^\tau x_j^\tau$ and $x_j^\tau x_k^\tau$ to $G$. Additionally
    adding a subdivision vertex to each of the edges $u^\tau x_i^\tau$, $u^\tau x_j^\tau$, and $u^\tau x_k^\tau$ and
    connecting them in the same order creates a rigid structure that fixes the order of $x_i^\tau$, $x_j^\tau$, and
    $x_k^\tau$ as required by the triplet. We call these four edges the \emph{triplet-constraint edges}. Unless stated
    otherwise, we ignore the existence of these subdivision vertices to avoid case distinctions. This way, we can treat
    all element vertices the same, but we implicitly presume the fixed order for element vertices constrained by
    triplet-constraint edges, as required by the corresponding triplet.
    
    Finally, for every pair $\{\varphi, \omega\} \in \binom{X}{3}$ of distinct subsets sharing two elements ${\varphi \cap
      \omega = \{x_a,x_b\}}$, create a new input graph $G^{\varphi, \omega}$ in the \SEFE instance \GU and add the two
    exclusive edges $x_a^\varphi x_b^\omega$ and $x_b^\varphi x_a^\omega$, called \emph{pairwise consistency edges}, to
    $G^{\varphi, \omega}$. This finishes the construction, which clearly works in polynomial time and ensures that the
    shared graph $G$ is connected and has maximum degree 4. It remains to prove the correctness of the construction.
    
    If $(X, \mathcal{T})$ is a yes-instance, let $\sigma$ denote the corresponding linear ordering of $X$. For every $\tau
    = \{x_i, x_j, x_k\} \in \binom{X}{3}$, let $(x_i, x_j, x_k)$ be the linear order of $\tau$ induced by $\sigma$. Since
    $u^\tau$ is a cutvertex of degree at most 3 in $G$, its rotation system does not affect the planarity of $G$; we pick
    an arbitrary order for the incident edges of~$u^\tau$. Embed the vertices around $v^\tau$ in the clockwise cyclic
    order $u^\tau, x_i, x_j, x_k, u^\tau$. If these vertices are not constrained by triplet-constraint edges, then all
    neighbors of $v_\tau$ except one have degree 1, the chosen order thus retains planarity. Otherwise, since $\sigma$
    satisfies all triplets of $\mathcal{T}$, the endpoints of all triplet-constraint edges appear consecutively around the
    corresponding vertex $v^\tau$ and the triplet-constraint edges thus do not introduce crossings. Overall, we therefore
    have a planar embedding \mcE of the shared graph $G$. We extend this embedding to an embedding of each input graph.
    Let $\{\varphi, \omega\} \in \binom{X}{3}$ be two distinct subsets sharing two elements $\varphi \cap \omega = \{x_a,
    x_b\}$. Without loss of generality we assume $x_a <_\sigma x_b$. Due to our choice of \mcE, this means that
    $x_a^\varphi < x_b^\varphi < x_a^\omega < x_b^\omega$ in the clockwise order around the face $f$ of \mcE containing
    all element vertices. Adding the two pairwise consistency edges $x_a^\varphi x_b^\omega$ and $x_b^\varphi x_a^\omega$
    corresponding to the pair $\{\varphi, \omega\}$ to face $f$ of \mcE therefore retains the planarity of the
    corresponding input graph. We therefore have a simultaneous embedding of~\GU.
    
    Conversely, assume that \GU is a yes-instance. Let $\mcEU$ denote a simultaneous embedding of \GU and let \mcE be the
    corresponding embedding of the shared graph. For every $\tau = \{x_i, x_j, x_k\} \in \binom{X}{3}$, the clockwise
    cyclic order of the four neighbors of $v^\tau$ in $\mathcal{E}$ induces a unique linear order $\sigma_\tau$ of $\tau$
    (e.g., the order $u^\tau, x_i^\tau, x_j^\tau, x_k^\tau, u^\tau$ induces the linear order $x_i <_{\sigma_\tau}  x_j
    <_{\sigma_\tau} x_k$). Note that the pairwise consistency edges ensure that all element vertices lie in the same face
    $f$ of \mcE. Let $x_a, x_b \in X$ denote two distinct elements of the ground set~$X$. For every pair $
    \varphi,  \omega \in \binom{X}{3}$ of distinct three-element subsets with $\varphi \cap \omega = \{x_a, x_b\}$, the two
    corresponding pairwise consistency edges ensure that $x_a^\varphi, x_b^\varphi, x_a^\omega, x_b^\omega$ alternate in
    the clockwise cyclic order around $f$ because otherwise the two edges would cross. For the linear orders
    $\sigma_\varphi$ and $\sigma_\omega$ of $\varphi$ and $\omega$, this implies that $x_a <_{\sigma_\varphi} x_b$ if and only
    if $x_a <_ {\sigma_\omega} x_b$. To construct an order $\sigma$ of $X$, we therefore set $x_a <_\sigma x_b$ if and only if $x_a
    <_{\sigma_\tau} x_b$ for some (and thus for all) $\tau \in \binom{X}{3}$ with $\{x_a, x_b\} \subset \tau$. The local
    ordering of every three-element subset of $X$ thus corresponds to the order of the same three elements in $\sigma$.
    Since the three-element subsets guarantee transitivity, $\sigma$ is a linear ordering of $X$. For every triplet of
    $t \in \mathcal{T}$, the triplet-constraint edges incident to the corresponding element vertices ensure that $\sigma$
    satisfies~$t$. Therefore, $(X, \mathcal{T})$ is a yes-instance of
    \textsc{Betweenness}.
\end{proof}
  This result can be further improved by eliminating all cycles of the shared graph and bounding the maximum degree of the union graph; see \Cref{fig:treeUnionDeg4}.
  
  \begin{figure}[t]
    \centering
    \includegraphics{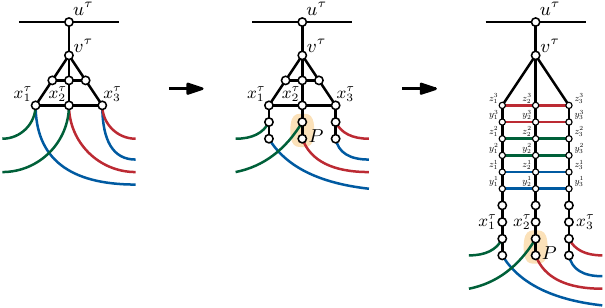}
    \caption{An example illustrating how the instance obtained from the reduction of \Cref{thm:npcMaxDeg4} can be
      modified such that each vertex has degree at most 4 in the union graph (first step) and the shared graph is acyclic
      (second step).}
    \label{fig:treeUnionDeg4}
  \end{figure}
  \begin{theorem}
    \label{thm:unionmaxdeg4}
    When $k$ is part of the input, \SEFE is NP-complete, even if the shared graph is a tree and the union graph has maximum degree 4.
  \end{theorem}
  \begin{proof}
    Let \GU be the \SEFE instance obtained by the reduction from \Cref{thm:npcMaxDeg4}; see \Cref{fig:boundedDegree} for
    an example. We first transform \GU into an equivalent instance that has maximum degree~4. Observe that only the
    element vertices of \GU can have degree higher than 4. For every $\tau \in \binom{X}{3}$ and for every $x_i \in
    \tau$, let $L$ denote the set of exclusive edges incident to $x_i^\tau$. Since $x_i^\tau$ is incident to at most
    three shared edges, it is $\deg(x_i^\tau) \leq |L| + 3$. In the shared graph, create a new path $P$ consisting of
    $|L|$ vertices and connect one of its endpoints to $x_i^\tau$; see the first step in \Cref{fig:treeUnionDeg4}.
    Disconnect all edges of $L$ from $x_i^\tau$ and connect them to the vertices of $P$ in such a way that each vertex
    of $P$ is incident to exactly one edge of $L$. Observe that $x_i^\tau$ and every vertex of $P$ now has degree at
    most 4 in the resulting instance \hatGU. It remains to show that \GU and \hatGU are equivalent.
    
    If \hatGU admits a simultaneous drawing $\hat{\Gamma^\cup}$, then contracting the shared path $P$ back into $x_i^\tau$
    clearly yields a simultaneous drawing of \GU.
    
    To show the other direction, assume that \GU admits a simultaneous drawing~$\Gamma^\cup$. As argued in the proof of
    \Cref {thm:npcMaxDeg4}, all element vertices and thus all pairwise consistency edges must be embedded in the same
    face $f$ of the shared graph in $\Gamma^\cup$. Draw path $P$ in face~$f$ next to $x_i^\tau$ and connect it to
    $x_i^\tau$. Since $L$ contains no two edges of the same input graph, the edges of $L$ may cross each other
    arbitrarily. We can therefore detach all edges of $L$ at $x_i^\tau$ and reattach them at the corresponding vertices of
    $P$ to obtain a simultaneous drawing of~\hatGU.
    
    Finally, we show that we can construct an equivalent instance $\hatGU_*$ where the shared graph is additionally
    acyclic. Observe that every simple shared cycle of $\GU$ (and consequently $\hatGU$) contains a triplet-constraint
    edge. Our goal is to replace the rigid structure induced by the triplet-constraint edges in the shared graph of
    $\hatGU$ with a rigid structure in every input graph. Let $\tau = \{x_i, x_j, x_k\} \in \binom{X}{3}$ be three
    elements fixed by a triplet $(x_i, x_j, x_k) \in
    \mathcal{T}$ and let $p$ denote the number of input graphs in \hatGU. To break all cycles, we first remove the
    triplet constraint edges and subsequently apply the following transformation to \hatGU; see the second step in
    \Cref{fig:treeUnionDeg4} for an example. For each $l \in [p]$, subdivide each of the edges $v^\tau x_i^\tau$, $v^\tau
    x_j^\tau$, and $v^\tau x_k^\tau$ twice. Let $(y_i^l, z_i^l)$, $(y_j^l, z_j^l)$, $(y_k^l, z_k^l)$ denote the respective
    pairs of subdivision vertices. Add the edges $y_i^l y_j^l$, $y_j^l y_k^l$, $z_i^l z_j^l$, and $z_j^l z_k^l$  as
    exclusive edges in the input graph \G{l}; we call these edges the \emph{exclusive triplet-constraint edges}. Note that
    these edges create a rigid structure in each input graph, fixing the order of the associated element vertices as
    required by the corresponding triplet.
    
    Assume that $\hatGU_*$ admits a simultaneous drawing $\hat{\Gamma_*}$. The exclusive triplet-constraint edges ensure
    the cyclic order $u^\tau < y_i^p < y_j^p< y_k^p < u^\tau$ or its reverse around $v^\tau$ in $\hat{\Gamma_*}$. Let
    $\hat{\Gamma}$ denote the drawing obtained from $\hat{\Gamma_*}$ by successively contracting the subdivision vertices
    back into $x_i^\tau$, $x_j^\tau$, and $x_k^\tau$ and subsequently removing the exclusive triplet-constraint edges. Note
    that contractions of shared edges retain the planarity of the drawing. Since these contractions do not alter the order
    of edges around $v^\tau$, it is $u^\tau < x_i^\tau < x_j^\tau < x_k^\tau$ or the reverse in $\hat{\Gamma}$. Since
    $v^\tau$ has no other incident edges, we can assume that the region spanned by the edges $x_i^\tau v^\tau$, $v^\tau
    x_j^\tau$, and the imaginary straight line $x_j^\tau x_i^\tau$ is not crossed by any edge in $\hat{\Gamma}$. Since the
    same holds for $x_j^\tau x_k^\tau$, we can simply draw the triplet-constraint edges of \hatGU as straight lines in
    $\hat{\Gamma}$ to obtain a simultaneous drawing of \hatGU.
    
    Conversely, given a simultaneous drawing $\hat{\Gamma}$ of \hatGU, the new subdivision vertices can be inserted
    directly next to $x_i^\tau$, $x_j^\tau$, and $x_k^\tau$ in $\hat{\Gamma}$ and the new edges can be drawn to take the
    same path as the (now removed) edges $x_i^\tau x_j^\tau$ and $x_j^\tau x_k^\tau$ in~$\hat{\Gamma}$. We thus obtain a
    simultaneous drawing of $\hatGU_*$.
  \end{proof}

  \subsection{Vertex Cover Number + Maximum Degree-1 Neighbors}
  \label{sec:shared-vc}
  
  Let $\Do$ denote the maximum number of degree-1 vertices adjacent to a single vertex in $G$
  and let $C$ denote a minimum vertex cover of $G$ with size $\varphi = \vc(G)$. In this section, we show that
  \SEFE is FPT with respect to $\varphi + \Do$.
  
  Since the shared graph is planar, the number of vertices with degree at least~3 in~$G$ is at most $3\varphi$ by
  \Cref{lm:planarDeg3}. Because $\Do$ bounds the number of degree-1 neighbors of each vertex in $G$, only the number of
  isolated vertices and degree-2 vertices in the shared graph remains unbounded. Unfortunately, it is difficult to reduce either of these
  vertex types in the general case.
  Instead, our goal is to enumerate all suitable embeddings of the shared graph.
Given the parameter $\varphi + \Do$, the only embedding choices we cannot afford to brute-force involve the P-nodes of
  $G$. Consider a set $U$ of at least three degree-2 vertices of $G$ that have the same neighborhood $\{u, v\}$. Note that
  neither of the vertices in $U$ has to be contained in the vertex cover $C$ of $G$ if both $u$ and $v$ are contained
  in~$C$, thus $|U|$ is not necessarily bounded by a function in~$\varphi + \Do$. In the SPQR-tree of $G$, $u$
  and $v$ are the poles of a P-node $\mu$ with at least $|U|$ parallel subgraphs and we therefore cannot afford to
  enumerate all permutations of~$\mu$. We call such a P-node \emph{two-parallel}. Even if we fix the embedding of $\mu$,
  the number of faces incident to its poles is not bounded, leaving too many embedding decisions if one of the poles of
  $\mu$ is a cutvertex in $G$; see \Cref{fig:two-parallel}.
  \begin{figure}[t]
    \centering
    \includegraphics[page=1]{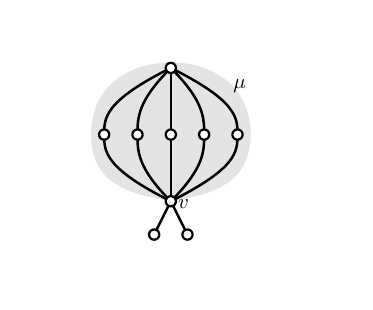}
    \caption{A two-parallel P-node $\mu$ with many possible embedding choices at cutvertex $v$.}
    \label{fig:two-parallel}
  \end{figure}

  To circumvent these issues, we show that we can limit the number of embedding decisions involving two-parallel P-nodes
  of $G$ using connectivity information of the union graph, allowing us to brute-force all suitable embeddings of each
  shared component in FPT-time. Subsequently, we can determine whether \GU admits a \SEFE with the given embeddings using
  the algorithm by Bläsius and Rutter \cite{Blaesius15} that solves \SEFE if every shared component has a fixed embedding.
  In other words, we use the bounded search tree technique where the branches of the search tree correspond to embedding
  decisions in the shared graph. We first branch for all possible embeddings of the blocks in $G$, and
  subsequently branch for all possible configurations of blocks around cutvertices.
  \subsubsection{Preprocessing Steps}
  To simplify the aforementioned problematic embedding decisions in the shared graph, we use three preprocessing steps
  introduced by Bläsius et al.~\cite{BlasiusKR18}\footnote{The proof of these preprocessing steps only considers the case
    of $k = 2$ input graphs. However, as the authors state in the introduction, all results extend to $k \geq 3$ input
    graphs in the sunflower case.}.

  First, we can decompose the split components of cutvertices of the union graph~\GU into independent instances and we can thus assume~\GU to be biconnected~\cite[Theorem~1]{BlasiusKR18}.
  The second preprocessing step guarantees that no separating pair of \GU separates a P-node of the shared graph~\cite[Lemma 4]{BlasiusKR18}.
  This connectivity information of \GU fixes the ordering of virtual edges of P-nodes in the shared graph up to reversal and thus heavily restricts the corresponding number of admissible embeddings~\cite[Theorem~2 + Lemma~5]{BlasiusKR18}.
  We remark that these preprocessing steps do not increase the parameter~$\varphi + \Do$.
  \begin{lemma}
    \label{lm:preprocessing}
    Given a \SEFE instance, we can compute in linear time an equivalent set of instances where each instance \GU satisfies the following properties:
    \begin{enumerate}[label=(\alph*)]
      \item\label{it:biconnected}  The union graph \GU is biconnected~\cite[Theorem 1]{BlasiusKR18}
      \item\label{it:sepPairUnion} If vertices $\{u, v\}$ are a separating pair of the union graph, then at most one of its split components contains a shared path between $u$ and $v$~\cite[Lemma 4]{BlasiusKR18}
      \item\label{it:PnodeBinary} The embedding of every P-node of the shared graph is a binary decision.~\cite[Theorem~2 + Lemma 5]{BlasiusKR18}
    \end{enumerate}
  \end{lemma}

  \subsubsection{Embedding the Blocks}
  We start by enumerating all suitable embeddings of the individual blocks in the shared graph~$G$. For this, it suffices
  to embed all P-nodes and R-nodes of the SPQR-forest of~$G$. Since R-nodes only allow a binary embedding decision and the
  same holds for all P-nodes in~$G$ by \Cref{lm:preprocessing}\ref{it:PnodeBinary}, we only need to bound the number these of P-nodes and
  R-nodes in our parameter. We do this with the following lemma.
  
  \begin{lemma}
    \label{lm:sumPQ}
    The number of P-nodes and R-nodes in the SPQR-forest of any planar graph~$G$ is in $O(\vc(G))$.
  \end{lemma}
  \begin{proof}
    Let $\mathcal{T}$ denote the SPQR-forest of $G$ and let $\varphi = \vc(G)$ denote the vertex cover number of $G$. We
    let $\nu(\mathcal{T})$ denote the sum of P-nodes and R-nodes in $\mathcal{T}$. First construct a graph $G'$ and its
    corresponding SPQR-forest $\mathcal{T'}$ by exhaustively removing all vertices of degree at most~1 from $G$ and
    additionally contracting each degree-2 vertex into one of its neighbors. We keep the parallel edges that may emerge
    after these contractions, thus $G'$ may be a multi-graph. Because we keep the parallel edges, note that
    $\mathcal{T}$ and~$\mathcal{T'}$ contain the same number of P-nodes and R-nodes, i.e., $\nu(\mathcal{T}) =
    \nu(\mathcal{T'})$.
    
    Since $G'$ only contains vertices of degree at least 3, and since there are at most as many such vertices in $G'$ as
    there are in $G$, $G'$ contains $O(\varphi)$ vertices by \Cref{lm:planarDeg3}. Because the edge contractions also
    retain the planarity of $G$, the number of bundles of parallel edges in $G'$ is bounded by the number of edges a
    planar graph with $O(\varphi)$ vertices can have, which is also $O(\varphi)$.
    
    Finally, let $G^*$ denote the graph obtained from $G'$ by replacing bundles of parallel edges with a single edge and
    let $T^*$ denote the corresponding SPQR-forest. Since $G'$ contains $O(\varphi)$ bundles of parallel edges,
    replacing the parallel edges can eliminate at most $O(\varphi)$ P-nodes. Observe that replacing parallel edges does
    not affect R-nodes, thus $\nu(\mathcal{T'}) - \nu(\mathcal{T^*}) = \nu(\mathcal{T}) -
    \nu(\mathcal{T^*}) \in O(\varphi)$. Since $G^*$ contains $O(\varphi)$ vertices and no parallel edges, planarity
    ensures that it also contains at most $O(\varphi)$ edges. 
Note that the number of inner nodes in an SPQR-tree is upper-bounded by the number of edges in the corresponding biconnected component.
    Since every edge is contained in exactly one biconnected component, the SPQR-forest $\mathcal{T^*}$ therefore contains $O(\varphi)$ inner
    nodes. As argued above, it holds that $\nu(\mathcal{T}) -
    \nu(\mathcal{T}^*) \in O(\varphi)$ and consequently $\nu(\mathcal{T}) \in O(\varphi)$.
\end{proof}
  
  Combining \Cref{lm:preprocessing}\ref{it:PnodeBinary} and \Cref{lm:sumPQ} allows us to enumerate all suitable embeddings of all blocks in $G$ in
  FPT-time, as there are at most $2^{O(\varphi)}$ such embeddings.

  \begin{corollary}
    \label{cor:fixedBlocks}
    Among all simultaneous embeddings of \GU, there are at most $2^{O(\varphi)}$ distinct combinations of embeddings of all blocks in the shared graph $G$.
    These embeddings can be enumerated in FPT-time.
  \end{corollary}
  
  \subsubsection{Nesting Blocks around Cutvertices}
  
  Given a fixed embedding for every block of the shared graph $G$, it remains to fix the order of incident edges at all cutvertices of $G$. Since this order is fixed for cutvertices of
  degree~2 and the total number of vertices of degree at least 3 in $G$ is at most $3\varphi$ by \Cref{lm:planarDeg3}, the
  number of relevant cutvertices is also at most $3\varphi$. Because we have already fixed the embedding of all blocks in
  $G$, fixing the order of incident edges at a cutvertex $v$ boils down to enumerating all
  suitable orderings and nestings of the blocks containing $v$.
  
  \begin{figure}[t]
    \centering
    \begin{subfigure}[t]{0.47\textwidth}
      \centering
      \includegraphics[scale=0.85,page=1]{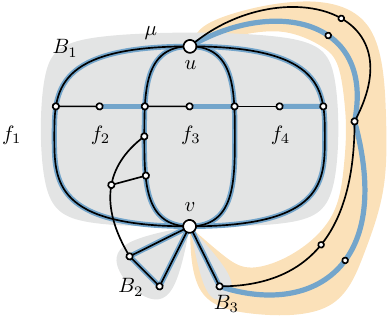}
      \caption{}
      \label{fig:CutvertexBlocksA}
    \end{subfigure}
    \hfill
    \begin{subfigure}[t]{0.47\textwidth}
      \centering
      \includegraphics[scale=0.85,page=2]{figures/MutableBlocks}
      \caption{}
      \label{fig:CutvertexBlocksB}
    \end{subfigure}
    \caption{(a) Three blocks $B_1$, $B_2$, and $B_3$ sharing the same cutvertex $v$. The block $B_1$ contains a
      two-parallel P-node $\mu$. Block $B_2$ can only be embedded in the faces $f_2$ and $f_3$ of $B_1$, thus it
      is a binary block with respect to $B_1$. The block $B_3$ is connected to the other pole $u$ of $\mu$ and can be
      embedded in all faces of $\mu$, thus it is a mutable block with respect to $B_1$. Note that $B_3$ is contained
      in a separate split component (highlighted in orange) with respect to the separating pair $\{u,v\}$ in the union
      graph, because $B_3$ is a mutable block. (b)~After assigning block $B_2$ to the face $f_3$ of $\mu$ (green arrow),
      $f_3$
      becomes occupied.}
    \label{fig:CutvertexBlocks}
  \end{figure}

  From now on, when we refer to faces of a block $B$ of $G$, we implicitly refer to the faces of the fixed embedding of
  $B$. For a cutvertex $v$ of $G$, let $B_1$ and $B_2$ denote two blocks of~$G$ containing $v$. Since $v$ may be a pole of
  a two-parallel P-node in $B_1$, the number of faces of $B_1$ that $B_2$ can be embedded into is in general not bounded
  by a function of $\varphi + \Do$; see \Cref{fig:CutvertexBlocksA}. We therefore cannot simply successively try to embed
  $B_2$ into every face of~$B_1$. However, since we may assume that the union graph \GU is biconnected
  (\Cref{lm:preprocessing}\ref{it:biconnected}), $B_2$ must be connected to $B_1$ in \GU via a path that is vertex-disjoint from $v$. This
  additional connectivity in \GU may already restrict the number of faces of $B_1$ that $B_2$ may be embedded into. If
  $B_2$ can be embedded into at most two faces of $B_1$, we say that $B_2$ is a \emph{binary block} with respect to $B_1$,
  otherwise it is a \emph{mutable block}; see \Cref{fig:CutvertexBlocks} for an example.
  
  If $B_2$ is a mutable block with respect to $B_1$, then $B_2$ must be connected to a single vertex $u$ of $B_1$ in $\GU
  - v$ and $\{u,v\}$ must be the poles of a P-node $\mu$ of $B_1$; see \Cref{fig:CutvertexBlocksA}. This
  also implies that $\{u, v\}$ is a separating pair in \GU and that $B_2$ and $B_1$ are contained in different split
  components with respect to~$\{u, v\}$. Let \SU denote the split component that $B_2$ is contained in and let \S{i}
  be the subgraph of \SU in \G{i} for all $i \in [k]$. Due to \Cref{lm:preprocessing}\ref{it:sepPairUnion}, we may assume that $\mu$ (and therefore $B_1$) is contained in a single split component with respect to $\{u,v\}$ in \GU.
  
  Our goal is to use the connectivity information of \SU to further restrict the number of faces of $B_1$ that \SU (and
  therefore $B_2$) can be embedded in. To this end, we adopt some notation from Bläsius et al.~\cite{BlasiusKR18}. Let $f$
  be a face of $\mu$ enclosed by the virtual edges $\varepsilon_1$ and~$\varepsilon_2$. We say that $f$ is
  \emph{$\circled{i}$-linked} for some $i \in [k]$, if there exists a path in $\G{i}  - \{u, v\}$ between $\exp (\varepsilon_1)$ and
  $\exp(\varepsilon_2)$ that is disjoint from $B_1$, except for its endpoints; see \Cref{fig:CommonEndsA}. Analogously, we
  say that the split component \SU is \emph{\circled{i}-connected}, if $u$ and $v$ are connected by a path in~\S{i}. For a
  set  $I \subseteq [k]$, we say that the face $f$ (resp. the split component \SU) is \emph{$I$-linked}
  (\emph{$I$-connected}), if $f$ is \circled{i}-linked (\circled{i}-connected) in \S{i} for all $i \in I$, but not for $i
  \in [k] \setminus I$; see \Cref{fig:CommonEndsB}. An $I$-linked face $f$ of $\mu$ is \emph{compatible} with a
  $J$-connected split component \SU if $I \cap J = \emptyset$. With the following lemma, we show that \SU must be
  compatible with some face of $\mu$, if $\GU$ is a yes-instance.
  
  \begin{figure}[t]
    \centering
    \begin{subfigure}[t]{0.38\textwidth}
      \centering
      \includegraphics[scale=1,page=1,trim={0.3cm 0 0.3cm 0}]{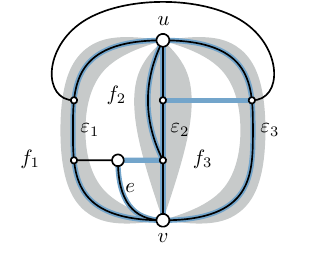}
      \caption{}
      \label{fig:CommonEndsA}
    \end{subfigure}
    \rulesep
\begin{subfigure}[t]{0.48\textwidth}
      \centering
      \includegraphics[scale=1,page=2]{figures/LinkInformation}
      \caption{}
      \label{fig:CommonEndsB}
    \end{subfigure}
    \rulesep
    \begin{subfigure}[t]{0.1\textwidth}
      \centering
      \includegraphics[scale=1,page=3,trim={2cm 0 1.5cm 0}]{figures/LinkInformation}
      \caption{}
      \label{fig:CommonEndsC}
    \end{subfigure}
    \caption{(a) A P-node $\mu$ of the shared graph with three virtual edges $\varepsilon_1$, $\varepsilon_2$, and
      $\varepsilon_3$. The faces $f_1$, $f_2$, and $f_3$ of $\mu$ are $\{2\}$-linked, $\emptyset$-linked, and $
      \{1\}$-linked, respectively. (b) Four split components with respect to the separating pair $\{u, v\}$ in \GU.
      The four components $S_1^\cup$, $S_2^\cup$, $S_3^\cup$, and $S_4^\cup$ are $\{1,2\}$-connected,
      $\{1\}$-connected, $\{2\}$-connected, and $\emptyset$-connected, respectively. $S_1^\cup$ is compatible with
      $f_2$, $S_2^\cup$ is compatible with $f_1$ and $f_2$, $S_3^\cup$ is compatible with $f_2$ and $f_3$, and
      $S_4^\cup$ is compatible with all faces of~$\mu$. (c) Although the split component \SU is compatible with
      the face $f_2$, \SU cannot be embedded into $f_2$, because the two shared edges $e$ and $g$ are
      conflicting.}
    \label{fig:CommonEnds}
  \end{figure}
  
  \begin{lemma}[{\cite[Lemma 5]{BlasiusKR18}}]
    \label{lm:compatible}
    Let \GU be an instance of \SEFE with a fixed embedding for every block of the shared graph. Let $\{u, v\}$ be a
    separating pair of \GU such that $u$ and $v$ are the poles of a P-node $\mu$ in $G$. If there exists a split
    component \SU with respect to $\{u, v\}$ in \GU that is not compatible with any face of $\mu$, then
    \GU is a no-instance.
  \end{lemma}
  \begin{proof}
    Let \SU be a split component of $\{u, v\}$. It suffices to show that \SU cannot be embedded into a face
    $f$ of $\mu$ that it is not compatible with.
    
    Since \SU is not compatible with $f$, there exists an $i \in [k]$ such that \SU is \circled{i}-connected and $f$ is
    \circled{i}-linked. Let $\varepsilon_1$ and $\varepsilon_2$ denote the two virtual edges of $\mu$ that enclose $f$.
    Because $f$ is \circled{i}-linked, there exists a path $\excl{P}{i}$ in $\G{i} - \{u, v\}$ between $\exp
    (\varepsilon_1)$ and $\exp(\varepsilon_2)$ that is disjoint from the block of the shared graph containing $\mu$,
    except for its endpoints. Since $\excl{P}{i}$ may not cross any other virtual edges of $\mu$, $\excl{P}{i}$ must be
    embedded in face $f$ in \G{i}.
Since \SU is \circled{i}-connected and therefore contains a path $L$ connecting $u$ and $v$ in \G{i},
    embedding
    \SU into face $f$ always introduces a crossing between $L$ and $P$. Consequently, there exists no simultaneous
    embedding of \GU where \SU is embedded in $f$.
  \end{proof}
  From now on, we can assume by \Cref{lm:compatible} that \SU is compatible with at least one face of $\mu$. Note that
  compatibility of \SU with a face $f$ of $\mu$ does not necessarily guarantee that \SU can be embedded into $f$; see
  \Cref{fig:CommonEndsC}.
  
  \paragraph{Embedding Split Components Into Compatible Faces.}

  We now show how the notion of compatibility allows us to limit the number of faces of $\mu$ each split component of the
  union graph (and therefore each mutable block) can be embedded in. We start by giving a high-level overview of our
  strategy. Let $B_1$ denote the block of the shared graph that the P-node $\mu$ is contained in. First, we assign a face
  of $B_1$ to each block that is a binary block with respect to $B_1$; we will show that the number of possible
  assignments is bounded by a function in $\varphi + \Do$. If a face $f$ of $\mu$ subsequently contains such a block
  incident to a pole $u$ or $v$ of $\mu$, we say that $f$ is \emph{occupied}; see \Cref{fig:CutvertexBlocksB}. For each
  remaining mutable block $B_2$ with respect to $B_1$ that can be embedded into any face of $\mu$, recall that there
  exists a split component \SU of the separating pair $\{u, v\}$ in \GU that $B_2$ is contained in (see
  \Cref{fig:CutvertexBlocksA} for an example). Let $\mathcal{F}_{\SU}$ be the set of faces of $\mu$ that are compatible
  with \SU. We will show that we can always decompose \SU into an independent instance if $\mathcal{F}_{\SU}$ contains a
  face that is not occupied. Finally, all remaining embedding choices with respect to the P-node $\mu$ only concern faces
  of $\mu$ that are occupied. Since we can bound the number of occupied faces and the number of blocks in $G$, we will
  finally show that we can brute-force all remaining embedding choices of the shared graph.
  
  As the first step, we now fix the position of all binary blocks. For every cutvertex $v$ in $G$ and for each pair $B_1,
  B_2$ of blocks containing $v$, we create two new branches if $B_2$ is a binary block with respect to $B_1$. The two
  branches correspond to the (at most) two faces of $B_1$ that $B_2$ can be embedded into. Since we have at most
  $3\varphi$ cutvertices with degree at least 3 in $G$ (\Cref{lm:planarDeg3}) and since we have at most $\varphi + \Do$
  blocks in $G$ that contain $v$ (each block contains either a vertex of the vertex cover, or a degree-1 vertex adjacent
  to $v$), we create at most $\BO {(2^{(\varphi + \Do)^2})^{3\varphi}} =
  \BO{2^{(\varphi + \Do)^2 \cdot 3\varphi}}$ branches to assign the binary blocks to all possible faces.
  
  Now consider two blocks $B_1$ and $B_2$ both containing cutvertex $v$, such that $B_2$ is a mutable block with respect
  to $B_1$. Recall that there exists a P-node $\mu$ with poles $v$ and $u$ in $B_1$, such that $B_2$ is contained in a
  different split component \SU of the separating pair $\{u, v\}$ in the union graph \GU (see
  \Cref{fig:CutvertexBlocksA} for an example). For a face $f$ of $\mu$, recall that we say that $f$ is occupied, if there
  exists a binary block incident to $u$ or $v$ that is assigned to the face $f$; see \Cref{fig:CutvertexBlocksB} for an
  example. Note that we need to be cautious around an occupied face $f$, since $f$ may contain additional shared edges
  incident to $u$ and~$v$, thus the problematic case shown in \Cref{fig:CommonEndsC} can occur and we have no guarantee
  that we can embed \SU into $f$, even if \SU is compatible with~$f$. If, however, \SU itself admits a
  \SEFE and is compatible with a face $f$ of $\mu$ that is not occupied, we now show that we can always embed \SU
  (and all other split components of $\mu$ that are compatible with $f$) into~$f$.
  
  \begin{lemma}[{[}Derived from {\cite[Lemma 6]{BlasiusKR18}}{]}]
    \label{lm:unoccupiedFaces} 
    Let \GU be an instance of \SEFE with a fixed embedding for every block of the shared graph $G$ and an assignment of all binary blocks to a face.
    Let $f$ denote an unoccupied face of a P-node $\mu$ of $G$ whose poles
    $u$ and $v$ are a separating pair of the union graph \GU. Let $\mathcal{S} \coloneqq \{\SU_1, \dots, \SU_t\}$ be the
    set of split components with respect to $\{u, v\}$ in \GU that are compatible with $f$. Then \GU admits a \SEFE if
    and only if
    \begin{enumerate}
      \item $\GU_\mu \coloneq \GU - \mathcal{S}$ admits a \SEFE, and
      \item for each $\SU_i \in \mathcal{S}$, $\SU_i + uv$ admits a \SEFE, where $uv$ is a shared edge.
    \end{enumerate}
  \end{lemma}
  \begin{proof} The proof is very similar to the proof of Lemma 6 in the paper by Bläsius et al.~\cite{BlasiusKR18}.
    They show that the decomposition works if none of the split components contain a shared edge incident to $u$ or $v$.
    Since we are interested in split components containing mutable blocks incident to $u$ or $v$, we unfortunately
    cannot exclude this case. However, we only have to additionally show that the fact that $f$ is unoccupied allows us
    to also consider split components that contain a shared edge incident to $u$ or $v$. Note that we may assume that no
    split component of $\mathcal{S}$ contains a shared path between $u$ and $v$ (\Cref{lm:preprocessing}\ref{it:sepPairUnion}).
    \Cref{fig:UnoccupiedProof} gives an illustration of the proof.
    
    If \GU admits a \SEFE, then the corresponding simultaneous embedding can be easily decomposed into simultaneous
    embeddings of $\GU_\mu = \GU - \mathcal{S}$ and $\SU_i + uv$ for each $\SU_i \in \mathcal{S}$.
    
    \begin{figure}[t]
      \centering
      \includegraphics{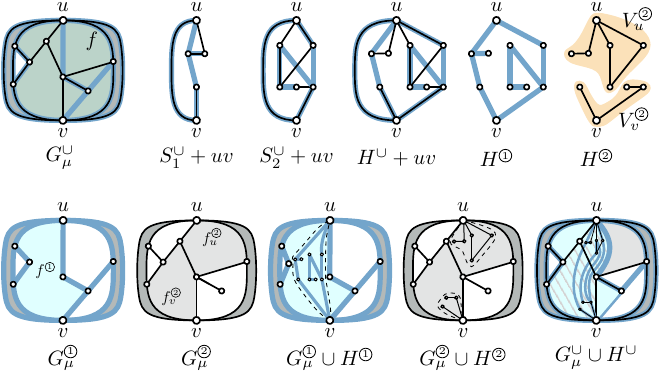}
      \caption{An example illustrating how split components $S^\cup_1$ and $S^\cup_2$ can be embedded in an unoccupied
        face $f$ of $G_\mu$ in \Cref{lm:unoccupiedFaces}. First, simultaneous embeddings of $S^\cup_1 + uv$ and
        $S^\cup_2 + uv$ are glued together to obtain a simultaneous embedding of $H^\cup + uv$. Because $H^\cup$ is
        \circled{1}-connected, there exists a face \excl{f}{1} of \excl{G_\mu}{1} containing both $u$ and $v$ where
        $\excl{H}{1}$ can be embedded in. Because $H^\cup$ is not \mbox{\circled{2}-connected}, the subsets $\excl{V_u}
        {2}$ and $\excl{V_v}{2}$ can be embedded into different faces $\excl{f_u}{2}$ and $\excl{f_v}{2}$ incident
        to $u$ and $v$, respectively. This yields a simultaneous embedding of $G^\cup_\mu \cup H^\cup = \GU$.}
      \label{fig:UnoccupiedProof}
    \end{figure}
    
    To show the other direction, assume that $\GU_\mu = \GU - \mathcal{S}$ and $\SU_i + uv$ (for each ${\SU_i \in
      \mathcal{S}}$) admit simultaneous embeddings respecting all embedding decisions for the shared graph we have made
    so far (embeddings of blocks and assignment of binary blocks to faces). We first glue the simultaneous embeddings of
    $\SU_1+uv,\dots,\SU_t+uv$ together in an arbitrary order to obtain a simultaneous embedding of $(S^\cup_1 \cup \dots
    \cup S^\cup_t) + uv$; see \Cref{fig:UnoccupiedProof} for an example. This is possible, because these graphs only
    share the vertices $u$ and $v$ and since each $S^\cup_i+uv$ contains the shared edge $uv$, we may assume that $u$
    and $v$ are on the outer face of each~$S^\cup_i$. For this reason, we can add $S_i^\cup$ to the outer face of
    $(S^\cup_1 \cup \dots \cup S^\cup_{i-1}) + uv$ to obtain a simultaneous embedding. We denote the resulting graph by
    $H^\cup := S^\cup_1 \cup \dots \cup S^\cup_t$ and the version of $H^\cup$ that additionally contains the shared edge
    $uv$ by $H^\cup + uv$.
    
    Since, by prerequisite of this lemma, every split component $\SU_i \in \mathcal{S}$ is compatible with face~$f$, the
    same also applies for $H^\cup$. Note that, in contrast to Lemma 6 in the paper by Bläsius et al.~\cite{BlasiusKR18},
    we explicitly allow $S^\cup_i$ (and therefore also $H^\cup$) to contain shared edges incident to $u$ and $v$. 
    However, we know that the simultaneous embedding of $\GU_\mu$ cannot contain shared edges incident to $u$
    and $v$ that are embedded into $f$, because (1) $f$ is unoccupied, thus no binary blocks have been assigned to $f$,
    and (2) all mutable blocks that are compatible with $f$ are contained in $\mathcal{S}$ and have thus been removed.
    In the following, we show that this allows us to embed $H^\cup$ into $f$, even if $H^\cup$ contains shared edges
    incident
    to $u$ or $v$, because the problematic case shown in \Cref{fig:CommonEndsC} can not occur.
    
    For every $j \in [k]$, for which $H^\cup$ is $\circled{j}$-connected (i.e., some split component $\SU_i$ is
    $\circled{j}$-connected), note that the compatibility criterion ensures that face $f$ is not \circled{j}-linked.
    This means that $\excl{G_\mu}{j}$ contains a face $\excl{f}{j}$ that is part of $f$ and contains both $u$ and $v$. 
    We embed $\excl{H}{j}$ into the face $\excl{f}{j}$ to obtain a planar embedding of $\excl{G_\mu}{j} \cup
    \excl{H}{j}$ (see $j = 1$ in \Cref{fig:UnoccupiedProof}).
    
    For every $j \in [k]$, for which $H^\cup$ is not $\circled{j}$-connected, first separate the vertices of $\excl{H}
    {j}$ into sets $\excl{V_u}{j}$ and $\excl{V_v}{j}$, such that $\excl{V_u}{j}$ contains all vertices of the connected
    component of $u$ and $\excl{V_v}{j}$ contains all remaining vertices. Pick arbitrary faces $\excl{f_u}{j}$ and~$\excl{f_v}{j}$ incident to $u$ and $v$ in $f$, respectively. We embed the vertices of $\excl{V_u}{j}$ and
    $\excl{V_v}{j}$ into $\excl{f_u}{j}$ and $\excl{f_v}{j}$, respectively, without changing the embedding of \excl{H}
    {i} (see $j = 2$ in \Cref{fig:UnoccupiedProof}).
    
    We therefore now have an embedding of $\GU_\mu \cup H^\cup = \GU$ where the embedding of each \G{i} is planar. It
    remains to show that it is a simultaneous embedding. Note that, in both of the above cases, we did not change the
    embeddings of $\GU_\mu$ and $H^\cup$. Therefore, all vertices except for maybe $u$ and $v$ have consistent edge
    orderings. Recall that the face $f$ is unoccupied and thus does not contain any shared edges incident to $u$ or $v$
    in~$G^\cup_\mu$. Since, for both vertices $u$ and $v$, we embedded all edges of each graph $\excl{H}{j}$ between the
    two shared edges belonging to face $f$, $u$ and $v$ therefore also have consistent edge orderings.
    Additionally, since all shared components of $H^\cup$ now belong to the face $f$ of the shared graph, and all shared
    components of $G_\mu^\cup$ lie on the outer face of $H^\cup$, the relative positions are also still consistent.
    Therefore, we obtain a simultaneous embedding of $\GU = G_\mu^\cup \cup H^\cup$.
  \end{proof}
  
  Consider once again a cutvertex $v$ of $G$ and two blocks $B_1$ and $B_2$ containing $v$, such that $B_2$ is a mutable
  block with respect to~$B_1$. Recall that this means that $B_1$ contains a P-node $\mu$ with poles $v$ and $u$, such that
  $B_2$ is contained in a split component $\SU$ of the separating pair $\{u,v\}$ in \GU. Let $\mathcal{F}_{S^\cup}$ denote
  the faces of $\mu$ that \SU is compatible with. \Cref{lm:unoccupiedFaces} now lets us assume that
  $\mathcal{F}_{S^\cup}$ only contains occupied faces of $\mu$. But since a face $f$ of $\mu$ is only occupied if we
  assigned a binary block incident to $u$ or $v$ to $f$, the number of occupied faces in $f$ is bounded by the number of
  blocks incident to $u$ and $v$. As argued before, there can be at most $\varphi + \Do$ blocks incident to a single
  vertex, thus $|\mathcal{F}_{S^\cup}| \leq 2(\varphi + \Do)$. We create a new branch for every face $f \in
  \mathcal{F}_{S^\cup}$ and assign the mutable block $B_2$ to $f$ in this branch. Thus, for every pair $B_1, B_2$ of
  blocks containing the same cutvertex $v$ we need at most $2(\varphi + \Do)$ branches to assign $B_2$ to every admissible
  face of~$B_1$. Since we have to do this for each of the $\BO{(\varphi + \Do)^2}$ pairs of blocks incident to a single
  cutvertex and since we have at most $3k$ cutvertices of degree at least~3, we need $\BO{((2(\varphi + \Do))^{
      (\varphi+\Do)^2})^{3\varphi}} =
  \BO{(2(\varphi + \Do))^{(\varphi+\Do)^2 \cdot 3\varphi}}$ branches to fix the nesting of the blocks at each cutvertex of
  $G$. Note that this bound also includes the branches we created to assign all binary blocks to their respective faces.
  
  \begin{corollary}
\label{cor:nesting} 
    Given a fixed embedding for every block of $G$, $\BO{(2(\varphi + \Do))^{(\varphi+\Do)^2 \cdot 3\varphi}}$ branches are sufficient to fix the nesting of the blocks around every cutvertex of $G$.
  \end{corollary}
  
  \subsubsection{Ordering Blocks Within the Same Face}
  
  We have now successfully fixed the nesting of the blocks at each cutvertex of $G$. Note, however, that there may be
  multiple blocks assigned to the same face $f$ of another block at a cutvertex $v$ of~$G$. In this case, we still need to
  determine the order of these blocks in the face~$f$. For this reason, we simply enumerate all possible orders of the
  blocks incident to every cutvertex of $G$ with degree at least 3. As argued before, a single vertex is contained in at
  most $\varphi + \Do$ blocks, hence there are at most $(\varphi + \Do)!$ orderings for the blocks around a single
  cutvertex. Since we have to do this at each of the at most $3\varphi$ cutvertices of degree at least 3
  (\Cref{lm:planarDeg3}), this creates an additional $\BO{((\varphi + \Do)!)^{3\varphi}}$ branches. Because every block of
  $G$ has a fixed embedding, we have now completely fixed the order of edges around each cutvertex of $G$.
  
  \begin{corollary}
    \label{cor:order} 
    $\BO{((\varphi + \Do)!)^{3\varphi}}$ branches are sufficient to fix the order of the blocks around
    every cutvertex of $G$.
  \end{corollary}

  \subsubsection{Putting Things Together}
  
  Finally, every connected component of the shared graph $G$ now has a fixed embedding in every branch, thus we can use
  the algorithm by Bläsius and Rutter \cite{Blaesius15} to determine whether \GU allows a simultaneous embedding with the
  given embeddings. Recall that our algorithm first fixes the embedding of each block in $G$ and subsequently nests and
  orders all blocks around cutvertices of $G$. By combining \Cref{cor:fixedBlocks,cor:nesting,cor:order}, we need a total of $\BO{2^{\BO{\varphi}} \cdot (2(\varphi + \Do))^{(\varphi+\Do)^2 \cdot 3\varphi}
    \cdot ((\varphi + \Do)!)^{3\varphi}}$
  branches to enumerate all admissible embeddings of all connected components in $G$. Note that not all branches lead to a
  valid embedding of all connected components in $G$ (e.g., different nesting decisions can contradict one another). If
  this is the case, we reduce to a trivial no-instance in the corresponding branch.
  
  \begin{theorem}
    \label{thm:vcDeg1}
    When $k$ is part of the input, \SEFE is FPT with respect to the maximum number of \mbox{degree-1} neighbors $\Do$ and the vertex cover number
    $\varphi$ of the shared graph.
\end{theorem}

  \subsection{Number of Cutvertices + Maximum Degree}
  \label{sec:cvMaxDeg}
  
  In this section, we consider \SEFE parameterized by the number of cutvertices $\cv(G)$ and the maximum degree $\Delta$
  of the shared graph $G$. This combination of parameters allows us to brute-force all possible orders of incident edges
  at cutvertices of $G$ and to branch for each combination of such orders. Given these fixed orders, we could, in theory,
  enforce the orders in the shared graph by replacing every cutvertex $v$ in $G$ with a wheel that fixes the order of its
  incident edges. In the resulting instance, every connected component of $G$ is biconnected, which is a restricted case
  of \SEFE that Bläsius et al.~\cite{BlasiusKR18} showed to be polynomial-time solvable.
  
  Unfortunately, adding such a wheel around a cutvertex $v$ in the shared graph may introduce crossings with exclusive
  edges incident to $v$ and even cause some of the input graphs to be non-planar. We circumvent this issue as follows. We
  first subdivide every edge incident to a cutvertex $v$ once, and subsequently add an additional input graph $\G{l}$ with
  $l \coloneqq k + 1$ where we fix the order of edges around cutvertices by connecting the subdivision vertices to a cycle
  in the desired order. In other words, at every cutvertex of the shared graph, we insert a wheel in \G{l}. In every
  simultaneous embedding of the union graph \GU, the edges around every cutvertex $v$ of $G$ now satisfy the fixed order,
  up to reversal. We therefore say that $v$ has a \emph{fixed rotation} in $G$.
  
  While every cutvertex of $G$ now has a fixed rotation, the shared graph itself unfortunately still contains cutvertices
  and we therefore cannot directly use the algorithm by Bläsius et al.~\cite{BlasiusKR18}. However, we use a similar
  approach as they do. They show that, for a biconnected planar graph $G$ and a vertex set $X$ of $G$ that shares a face
  in some planar embedding of $G$, the order of $X$ along any simple cycle of $G$ is fixed up to reversal. This
  essentially allows them to decompose the instance at such cycles, as the ``interface'' the graph provides for the
  vertices of $X$ is independent from the embedding.
  
  Since facial cycles of non-biconnected graphs are not necessarily simple, the same approach does not generalize to non-biconnected graphs, because changing the rotation of cutvertices may also change
  the order of vertices in the cycle. In the following, we show that the statement also holds for non-biconnected graphs
  if every cutvertex has a fixed rotation. 
We say that a cycle $C$ of $G$ is \emph{face-embeddable}, if there exists a planar embedding of $G$ where $C$ is a facial cycle.
  
  \begin{lemma}[Derived from {\cite[Lemma 8]{BlasiusKR18}}]
    \label{lm:cycleFixedOrder}
    Let $G$ be a planar graph where every cutvertex has a fixed rotation. Let $X$ be a set of vertices that are
    incident to a common face in some planar embedding of $G$. Then the order of $X$ in any face-embeddable cycle of $G$
    containing~$X$ is unique up to reversal.
\end{lemma}
  \begin{proof}

    We prove the statement via induction on the number of blocks in $G$.
    As a base case,~$G$ only consists of a single block and is therefore biconnected.
    Since a face of a biconnected graph always induces a simple facial cycle, every face-embeddable cycle of $G$ is simple and therefore the statement follows from the lemma by Bläsius et al.~\cite[Lemma 8]{BlasiusKR18}.
    
    Now consider the case where $G$ consists of $k$ blocks.
    Let $C$ be an arbitrary face-embeddable cycle of $G$.
    Pick a block $B$ of $G$ that is only incident to a single cutvertex~$v$, let $G' \coloneqq G - (B - v)$ be the subgraph of $G$ obtained after removing $B$ (but not $v$) and let $C'$ and $C_B$ be the cycles induced by $C$ in $G'$ and $B$, respectively.
    Since $G'$ and $B$ both contain less than $k$ blocks, the order in which $C'$ and $C_B$ visit the vertices of $X \cap C'$ and $X \cap C_B$, respectively, is unique up to reversal by the inductive hypothesis.
    Note that, in any planar embedding of a graph, the order in which a facial cycle visits the edges incident to a single cutvertex $v$ is unique up to reversal.
    Since $v$ has a fixed rotation in our case, this order is unique (up to reversal) among all planar embeddings of $G$.
    Because $C'$ and $C_B$ each visit at least two edges incident to $v$, the order in which $C$ visits the vertices of $X$ in $G$ is therefore unique up to reversal, which concludes the proof.
  \end{proof}
  
  Note that a necessary condition for planarity is that the edges of two distinct blocks containing the same cutvertex $v$
  may not alternate around $v$. If the fixed rotation of $v$ violates this condition, we report a no-instance in the
  corresponding branch. Roughly speaking, the following \Cref{lm:cvDecomposition} tells us that different split components
  of a cutvertex with fixed rotation interact very nicely. The fixed rotation of $v$ already determines the nesting of its
  split components, and the connectivity in the union graph determines the set $X$ of \emph{attachment vertices} of a
  split component $S_1$ that must share a face in $G$ with a different split component $S_2$ in every simultaneous
  embedding; these are the vertices of $S_1$ that are connected to $S_2$ in $\GU - v$ via a path that is vertex-disjoint
  from $S_2$ except for its endpoints. By \Cref{lm:cycleFixedOrder}, the order of the attachment vertices $X$ on this
  facial cycle of $S_1$ is essentially unique and independent from the embedding of $S_1$. The vertices of $X$ thus
  essentially act as a fixed interface between $S_1$ and $S_2$. If the subgraph of
  \GU induced by $S_1$ admits a any simultaneous embedding where all vertices of $X$ lie on the outer face, we can
  thus subsequently assume that $S_1$ has this fixed embedding in \GU. We will use this observation to decompose
  split components at cutvertices of the shared graph to receive smaller instances of \SEFE.
  
  Let $v$ be a cutvertex of $G$ with a fixed rotation system $\sigma_v$. A \emph{block-preserving cut} of $v$ partitions
  the vertices of $G - v$ into two non-empty sets $V_1$ and $V_2$, where (i) the edges incident to $v$ that are also
  incident to a vertex of $V_1$ are consecutive in $\sigma_v$ and (ii) for every block $B$ of $G$, the vertices of $B$
  (ignoring $v$) are either all contained in $V_1$ or all contained in $V_2$; see \Cref{fig:Decomposition} for an example.
  We remark that the choice of $V_1$ and $V_2$ is in general not unique.
  Note that $\sigma_v$ induces a nesting for all split components of $v$ contained in $V_1$ (respectively $V_2$). We
  consider two such split components $S_1$ and $S_2$ as equivalent, if $S_1$ is nested within $S_2$ or vice versa. To
  obtain a graph $\GU_1$ from \GU, contract each set of equivalent split components of $v$ contained in $V_2$ into a
  single vertex adjacent to $v$. We let $Y_2$ denote the set containing these contracted vertices. These contractions may result in edges
  that appear in multiple graphs but are not shared. We keep these parallel edges by implicitly assuming that each of them
  contains a subdivion vertex. For each vertex $x \in Y_2$, replace the shared edge $vx$ with an exclusive edge $vx$ in
  \G{l}. Finally, contract every remaining connected component of set~$V_2$ into a single vertex.
  To obtain the graph $\GU_2$ with contraction vertices $Y_1$, symmetrically apply the same procedure for $V_2$ on~\GU.
  Since $v$ has fixed rotation and forms a wheel with its adjacent edges in $\G{l}$, its rotation remains fixed in $\GU_1$
  and $\GU_2$, respectively, after the contractions; see \Cref{fig:Decomposition}.
  
  Given simultaneous embeddings of both $\GU_1$ and $\GU_2$, let $\mcE_1$ and $\mcE_2$ be the corresponding embeddings of
  the shared subgraphs $G_1 \coloneqq G[V_1 \cup \{v\}]$ and $G_2 \coloneqq G[V_2 \cup \{v\}]$, respectively. The
  \emph{compatibility embedding} of these two simultaneous embeddings is an embedding of the shared graph $G$ obtained by
  merging the embeddings $\mcE_1$ and $\mcE_2$ at $v$ such that $v$ satisfies the rotation $\sigma_v$. We can show that
  the existence of a a simultaneous embedding of \GU where the shared graph has the compatibility embedding is equivalent
  to the existence of any simultaneous embedding.
  
  \begin{figure}[t]
    \centering
    \includegraphics{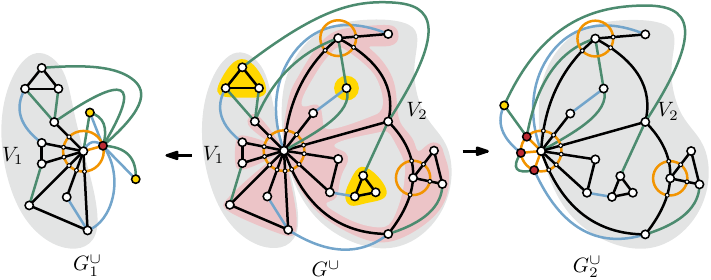}
    \caption{The graphs $\GU_1$ and $\GU_2$ corresponding to the block-preserving cut $\{V_1, V_2\}$ of $v$ in \GU. The
      red vertices are the vertices resulting from contracting the nested split components of $v$. The edges of graph $\G
      {l}$ fixing the rotation of all cutvertices using wheels are drawn in orange.}
    \label{fig:Decomposition}
  \end{figure}
  
  \begin{lemma}
    \label{lm:cvDecomposition}
    Let \GU be a \SEFE instance with a fixed rotation for all cutvertices of $G$ and let $v$ be a cutvertex of $G$. Then
    \GU admits a simultaneous embedding if and only if, for any block-preserving cut of $v$, (i) the corresponding graphs
    $\GU_1$ and $\GU_2$ admit simultaneous embeddings $\mcEU_1$ and $\mcEU_2$, respectively, and (ii) \GU admits a
    simultaneous embedding where the shared graph has the compatibility embedding of $\mcEU_1$ and $\mcEU_2$.
  \end{lemma}
  \begin{proof}
    We first show the reverse direction. If both conditions are satisfied, then the second condition immediately yields
    the desired simultaneous embedding of~\GU. By the definition of the compatibility embedding, this embedding satisfies
    the fixed rotation of $v$.
    
    Conversely, assume that we have a simultaneous drawing $\Gamma^\cup$ of \GU. We first show that Property (i) holds.
    Since $v$ has a fixed rotation $\sigma_v$ and the nested split components we contract are consecutive with respect to
    $\sigma_v$, these contractions retain the planarity of the drawing. The same applies for the contractions of the other
    shared components. In the resulting drawing, we only have to replace shared edges with exclusive edges in \G{l} to
    obtain a simultaneous drawing of $\GU_1$ (respectively $\GU_2$).
    
    For Property (ii), consider a simultaneous embedding $\mcEU$ of \GU and simultaneous embeddings $\mcEU_1$ and
    $\mcEU_2$ of $\GU_1$ and $\GU_2$, respectively, obtained from some block-preserving cut of $v$ (Property (i)). Let
    $C_1$ (respectively $C_2$) be the connected component of $G_1$ ($G_2$) containing~$v$. Let the attachment vertices $X$
    denote the set of vertices of $C_1$ that are connected to a vertex of $C_2$ in \GU via a path that is vertex disjoint
    from $C_1$ except for its endpoints. Since the vertices of $C_2 - v$ are connected in $\GU - C_1$ (due to the wheel in
    \G{l}), all vertices of $X$ must be incident to a common face $f$ of $C_1$ in \mcEU. By \Cref{lm:cycleFixedOrder}, the
    order of $X$ on this facial cycle is unique up to reversal. Since the same holds for the vertices of $X$ on the face
    $f_1$ of $C_1$ containing the contraction vertices $Y_2$ in $\mcEU_1$, and since $v$ has a fixed rotation, we can
    assume that $X$ has the same order on the facial cycles of $f$ and $f_1$ in $\mcEU$ and $\mcEU_1$, respectively. We
    can therefore replace the embedding $\mcEU[C_1]$ in \mcEU with its corresponding embedding $\mcEU_1[C_1]$ from
    $\mcEU_1$ while retaining the planarity of \mcEU. All other connected components of $G_1 - C_1$ can be fixed in
    $\mcEU$ with their embedding of $\mcEU_1$ using the same argumentation. Applying the same procedure for the subgraph
    $G_2$ using $\mcEU_2$, we thus obtain a simultaneous embedding of \GU where the shared graph has the compatibility
    embedding.
  \end{proof}
  
  Using \Cref{lm:cvDecomposition}, we can solve our \SEFE instance \GU with a fixed rotation for every cutvertex
  of $G$ as follows. First, pick an arbitrary cutvertex $v$ of $G$ and an arbitrary block-preserving cut of $v$. As long
  as $G$ contains a cutvertex, we always find such a cut where we can make progress towards biconnectivity.
  Decompose the instance into the instances $\GU_1$ and $\GU_2$ as described above and solve them recursively. Note that
  our decomposition ensures that $\GU_1$ and $\GU_2$ still have a fixed rotation for all cutvertices due to the
  wheels in graph~\G{l}. If $\GU_1$ or $\GU_2$ is a no-instance, we can report that \GU is a no-instance by \Cref
  {lm:cvDecomposition}. Otherwise, we obtain a pair of simultaneous drawings, and we test whether \GU admits a \SEFE with
  the corresponding compatibility embedding as a fixed embedding of the shared graph using the algorithm by Bläsius and
  Rutter~\cite{Blaesius15}, which solves \SEFE in quadratic time if every connected component of the shared graph has a
  fixed embedding. By \Cref{lm:cvDecomposition}, the result of this test determines whether \GU is a yes-instance
  of \SEFE.
  
  At the base case of the recursion, every connected component of the shared graph $G$ is biconnected. To solve this
  restricted case in polynomial time, we can use the algorithm by Bläsius et al.~\cite{BlasiusKR18}. Note that every
  shared block of the original instance appears in exactly one base case of the recursion, we thus have a linear number of
  instances that need to be solved, each of them has linear size. Together with the running time of
  $O((\Delta!)^{\cv(G)})$ that was necessary to fix the rotation of all cutvertices in the shared graph, we finally obtain
  the following result.
  \begin{theorem}
    \label{thm:CvMaxDeg}
    When $k$ is part of the input, \SEFE is FPT with respect to the number of cutvertices $\cv(G)$ and the maximum degree $\Delta$ of $G$ and
    can be solved in ${O((\Delta!)^{\cv(G)} \cdot n^{O(1)})}$ time.
  \end{theorem}
  
  \section{Conclusion}
  \label{sec:conclusion}

  Our FPT algorithms for the vertex cover number (\Cref{sec:vc-union}) and the feedback edge set number (\Cref{sec:fesunion}) of the union graph both additionally require the number of input graphs~$k$ as a parameter, raising the question whether similar results can be obtained for unbounded~$k$.
  In contrast, the reduction showing para-NP-hardness of \SEFE with respect to the twin cover number of the union graph (\Cref{sec:twincoverunion}) requires~$k$ to be unbounded.
  Is the problem FPT parameterized by the twin cover number plus~$k$?
  
  For the shared graph, we have shown that all eight parameters we considered (and most of their combinations; see \Cref
  {fig:Table1}) are intractable. This even includes the vertex cover number, one of the strongest graph parameters that
  upper bounds numerous other metrics like the feedback vertex set number, the treedepth, the pathwidth, and the treewidth. Most of the
  intractabilities regarding degree-related parameters follow from our hardness reduction for \SEFE with maximum degree~4
  in \Cref{thm:unionmaxdeg4}. However, \Cref{thm:unionmaxdeg4} only holds if the number $k$ of input graphs is unbounded, as the reduction inherently requires many input graphs. It is therefore an interesting question whether these
  parameters become tractable when one additionally considers $k$ as a parameter.

  \bibliography{bibliography}
  
\end{document}